\documentclass{article}

\usepackage{graphicx}
\usepackage{color}
\usepackage{amssymb,amsmath,amsthm}
\usepackage{graphicx}
\usepackage{cite}
\usepackage{fullpage}

\newtheorem{definition}{Definition}

\newtheorem{theorem}[definition]{Theorem}

\newtheorem{lemma}[definition]{Lemma}

\newtheorem{corollary}[definition]{Corollary}

\begin{document}

\title{An Optimal Single-Path Routing Algorithm\\ in the Datacenter Network DPillar}

\author{Alejandro Erickson, Iain A. Stewart\\School of Engineering and Computing
  Sciences, Durham University,\\South Road, Durham DH1 3LE, U.K. \and Abbas Eslami Kiasari, Javier Navaridas\\School of Computer Science, University of Manchester,\\Oxford Road, Manchester M13 9PL, U.K.}

\date{}

\maketitle

\begin{abstract}
DPillar has recently been proposed as a server-centric datacenter network and is combinatorially related to (but distinct from) the well-known wrapped butterfly network. We explain the relationship between DPillar and the wrapped butterfly network before proving that the underlying graph of DPillar is a Cayley graph; hence, the datacenter network DPillar is node-symmetric. We use this symmetry property to establish a single-path routing algorithm for DPillar that computes a shortest path and has time complexity $O(k)$, where $k$ parameterizes the dimension of DPillar (we refer to the number of ports in its switches as $n$). Our analysis also enables us to calculate the diameter of DPillar exactly. Moreover, our algorithm is trivial to implement, being essentially a conditional clause of numeric tests, and improves significantly upon a routing algorithm earlier employed for DPillar. Furthermore, we provide empirical data in order to demonstrate this improvement. In particular, we empirically show that our routing algorithm improves the average length of paths found, the aggregate bottleneck throughput, and the communication latency. A secondary, yet important, effect of our work is that it emphasises that datacenter networks are amenable to a closer combinatorial scrutiny that can significantly improve their computational efficiency and performance.

\end{abstract}

\section{Introduction}\label{sec:intro}

Datacenters are assuming an increasingly important role in the global computational infrastructure. They provide platforms for a wide range of data-intensive applications and activities including web search, social networking, online gaming, large-scale scientific deployments and service-oriented cloud computing. There is an increasing demand that datacenters incorporate more and more servers, and do so in a cost-effective fashion, but still so that the resulting platform is computationally efficient (in various senses of the term).

A \emph{datacenter network\/} (DCN) comprises the physical communication infrastructure underpinning a datacenter. One of the main aspects of a datacenter network is the \emph{topology\/} by which the servers, switches and other components of the datacenter are interconnected; the choice of topology strongly influences the datacenter's practical performance (see, \emph{e.g.}, \cite{LMV13}). For simplicity, henceforth by DCN we refer to the datacenter network topology. Originally, DCNs were hierarchical with expensive core routers that became bottlenecks in terms of both performance and cost. They evolved into tree-like, \emph{switch-centric\/} DCNs, built from commodity-off-the-shelf (COTS) components; that is, so that the servers are located at the `leaves' of a tree-like structure that is composed entirely of switches and where the routing intelligence resides within the switches. Such DCNs can offer better load balancing capabilities and so are less prone to bottlenecks but have limited scalability due to (the size of) routing tables within the switches. Typical examples of such switch-centric DCNs are ElasticTree \cite{HSM06}, Fat-Tree \cite{ALV08}, VL2 \cite{GHJ09}, HyperX \cite{ABD09}, Portland \cite{MPF09} and Flattened Butterfly \cite{AMW10}. 

Alternative architectures have recently emerged and  \emph{server-centric\/} DCNs have been proposed whereby the interconnection intelligence resides within the servers as opposed to the switches. Now, switches only operate as dumb crossbars (and consequently the need for high-end switches is diminished as are the infrastructure costs). This paradigm shift means that more scalable topologies can be designed and the fact that routing resides within servers, which are easier to program than are switches, means that more effective routing algorithms can be adopted. However, server-centric DCNs are not a panacea as packet latency can increase, with the need to handle routing providing a computational overhead on the server. Typical examples of server-centric DCNs are DCell \cite{GWT08}, BCube \cite{GLL09}, FiConn \cite{LGW09}, CamCube \cite{ACR10}, MCube \cite{WWY10}, DPillar \cite{LYY12}, HCN and BCN \cite{GCL13} and SWCube, SWKautz, and SWdBruijn \cite{LW15}. An additional positive aspect of some server-centric DCNs is that not only can commodity switches be used to build the datacenters but commodity servers can too; the DCNs FiConn, MCube, DPillar, HCN, BCN, SWCube, SWKautz, and SWdBruijn are all such that any server only needs two NIC ports (the norm in commodity servers) in order to incorporate it into the DCN.

It is with the DCN DPillar that we are concerned here. DPillar is an established
and one of the most promising benchmark dual-port server-centric DCNs. Moreover,
DPillar is one of the even fewer dual-port server-centric DCNs for which no
server-node is adjacent to any other server-node, the others being SWKautz,
SWCube, and SWdBruijn. DPillar has recently been compared with other dual-port
server-centric DCNs \cite{LW15}. It was shown that when the diameter of the DCN
is normalized, DPillar can incorporate more servers than FiConn and BCN, a
similar number of servers to SWCube, and (usually) less servers than SWKautz and
SWdBruijn. However, DPillar, SWCube, SWKautz, and SWdBruijn were shown to have
similar bisection widths and all have better bisection widths than FiConn and
BCN. Whilst SWCube, SWKautz, and SWdBruijn were compared with each other in
\cite{LW15} with regard to aspects of routing in relation to fault-tolerance and
handling congestion, there was no comparison of these three DCNs with DPillar.
Such an evaluation is currently missing and would obviously be tied to a
particular routing algorithm for DPillar, an observation that we will return to
in a moment.

As we shall see, DPillar is essentially obtained by replacing complete bipartite subgraphs $K_{\frac{n}{2},\frac{n}{2}}$ in a wrapped butterfly network (see, \emph{e.g.}, \cite{Lei92}) with a switch with $n$ ports. In \cite{LYY12}, basic properties of DPillar are demonstrated and single-path and multi-path routing algorithms are developed (along with a forwarding methodology for the latter). Our focus here is on single-path routing (also known as single-source deterministic routing). The algorithm in \cite{LYY12} is appealing in its simplicity but for most source-destination pairs it does not produce a path of shortest length; indeed, there is often a significant discrepancy between the lengths of the path produced by the algorithm in \cite{LYY12} and a shortest path (as we demonstrate later). We remedy this situation and develop a single-path routing algorithm that always outputs a shortest path. Although the proof of correctness of our algorithm is non-trivial, the actual algorithm itself is a very simple sequence of numeric tests and has the same time complexity as the original single path routing algorithm, \emph{i.e.}, linear in the number of columns within DPillar. 

Furthermore, we undertake an empirical evaluation and show that according to our experiments, the original single path routing algorithm for DPillar from \cite{LYY12} fails to provide a shortest path route for more than 51\% and up to 78\% of the server pairs; this translates into our algorithm giving an improvement in the range of 20-30\% in terms of the average path length derived. Note that a reduction in path length not only means that the latency of the network traffic will be reduced (between 20 and 25\%, in our experiments), but also that as less resources are required for transmitting data, the overall throughput of the network should also increase. To verify this latter contention, we empirically measure the aggregate bottleneck throughput (the most widely accepted datacenter throughput metric) for both algorithms and we find that our algorithm yields improvements in the range of 25-120\%, with a mean of 65\% and a median of 75\%. The substantial improvements in average path length and throughput, together with the algorithmic simplicity of our proposal, more than motivates its utilization in production systems. As by-products of the development of our algorithm, we prove that the DCN DPillar is, in essence, a Cayley graph, and thus node-symmetric (that is, there is an automorphism mapping any server to any other server), and we obtain the diameter of the DCN DPillar exactly. 

Let us now return to our earlier remark as regards the current lack of a comparison in the literature of DPillar with SWCube, SWKautz, and SWdBruijn with respect to aspects of routing in relation to fault-tolerance and handling congestion. Were we to embark on this comparison prior to the results of our paper then we would be doing a disservice to DPillar as we would be working with the routing algorithm from \cite{LYY12} which we prove (and empirically validate) here to be significantly worse in all respects than the routing algorithm we develop in this paper. We intend in future to undertake an extensive evaluation of aspects of routing for dual-port server-centric DCNs including DPillar, SWCube, SWKautz, and SWdBruijn but thanks to the results of this paper, this will now be with respect to our improved routing algorithm for DPillar (of course, such an evaluation is beyond the scope of this paper).

In the next sections, we give an explicit definition of the DCN DPillar, both
algebraically and as a derivation from wrapped butterfly networks, before
showing how to abstract DPillar as a directed graph and proving that the
resulting directed graph is a Cayley graph; an immediate consequence is that the
DCN DPillar is node-symmetric. In Section~\ref{sec:routing}, and using the
newfound property of node-symmetry, we explain how solving the single-path
routing problem in our abstraction of DPillar can be further abstracted so that
it is equivalent to a routing problem in what we call a marked cycle, and in
Section~\ref{sec:routingmarked} we prove that shortest paths in this marked
cycle must have severe restrictions on their structure. We use these
restrictions to develop our single-path routing algorithm for DPillar in
Section~\ref{sec:algorithm} and establish its correctness and its time
complexity. To support our theoretical analysis, we provide empirical evidence
that the length of the (shortest) path obtained by our single-path routing
algorithm is significantly shorter than the length of the path obtained by the
single-path routing algorithm from \cite{LYY12} for many source-destination
pairs, and we calculate the diameter of DPillar explicitly. Our conclusions and
directions for further research are given in
Section~\ref{sec:conclusions}\footnote{Some results from this paper appeared in
  preliminary form in: A. Erickson, A. Kiasari, J. Navaridas and I.A. Stewart,
  An efficient shortest path routing algorithm in the data centre network
  DPillar, \emph{Proc. of 9th Ann. Int. Conf. on Combinatorial Optimization and
    Applications}, 2015, pp. 209--220; some proofs and results were omitted and
  there was no experimental evaluation.}.

\section{The DCN DPillar}\label{sec:def}

In this section, we explicitly define the DCN DPillar and explain how the DCN DPillar can be (informally) constructed from a wrapped butterfly network.

\subsection{A definition of DPillar}

The DCN DPillar \cite{LYY12} consists of a collection of switches, each of which has $n$ ports, with $n\geq 2$ even, and a collection of servers, each of which has $2$ NIC ports. The names of the servers are $\{(c,v_{k-1}v_{k-2}\ldots v_0): 0\leq c \leq k-1; 0\leq v_i \leq \frac{n}{2}-1; 0\leq i\leq k-1\}$ where $k\geq 2$ (we refer to $k$ as the \emph{dimension\/}): the first parameter, $c$, is the \emph{column-index\/} and denotes the \emph{column\/} in which the server resides, whilst the second parameter $v_{k-1}v_{k-2}\ldots v_0$ is the \emph{row-index\/} and denotes the server's position within a column (from the left, the bit positions are $k-1, k-2, \ldots,0$; note that we refer to the values as `bits' and their positions as `bit' positions). We denote the DCN DPillar with parameters $n$ and $k$, as above, by DPillar$_{n,k}$. Consequently, DPillar$_{n,k}$ has $k(\frac{n}{2})^k$ servers. 

We term the collections of servers `columns' as we visualize the servers within a column as being stacked vertically within that column, with the row-indices of the servers, from top to bottom, being given in increasing lexicographic order on $\{0,1,\ldots,\frac{n}{2}-1\}^k$; so, if $n=6$ and $k=4$, for example, then the ordering is given by $0000<0001<0002<0010<0011<0012<0020<\ldots$ and so on. There are $(\frac{n}{2})^{k-1}$ switches located between column $i$ and column $i+1$, for $i=0,1,\ldots,k-2$, and also between column $k-1$ and column $0$; thus, there are $k(\frac{n}{2})^{k-1}$ switches in DPillar$_{n,k}$. We think of the switches between two columns of servers as appearing in a column too, with the names of the switches in a column being $\{0,1,\ldots,\frac{n}{2}-1\}^{k-1}$ and again stacked from top to bottom in increasing lexicographic order. If a switch lies between server-column $c$ and server-column $c+1$, where $c\in\{0,1,\ldots,k-1\}$ and addition is modulo $k$, then we say that its \emph{column\/} is column $c$ (henceforth, we assume that addition and subtraction on the names of columns are always modulo $k$). The columns of servers and switches for DPillar$_{6,3}$ can be visualized as in Fig.~\ref{DPillar} (note that the servers in the right-most and left-most columns are identical but are shown separately to facilitate visualization).

\begin{figure*}[t]
\centering
\scalebox{1.0}[1.0]{
\includegraphics{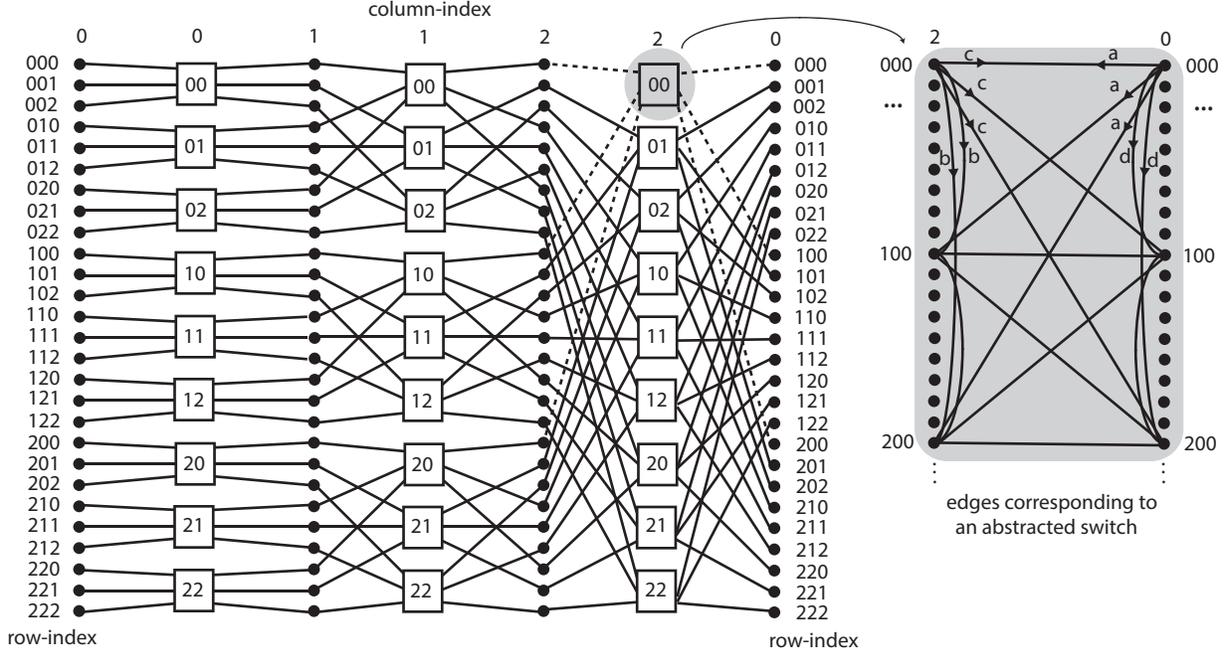}}
\caption{Visualizing DPillar$_{6,3}$.}\label{DPillar}
\end{figure*}

All links are server-switch links and are from a server in (server-)column $c$ to a switch in (switch-)column $c$ or from a server in (server-)column $c+1$ to a switch in (switch-)column $c$ (where $c\in\{0,1,\ldots,k-1\}$). Let $(c,v_{k-1}v_{k-2}\ldots v_0)$ be a server in column $c$. The switch to which it is connected in column $c$ is the switch named $v_{k-1}\ldots v_{c+1}v_{c-1}\ldots v_0$. If $(c+1,v_{k-1}v_{k-2}\ldots v_0)$ is a server in column $c+1$ then the switch to which it is connected in column $c$ is the switch named $v_{k-1}\ldots v_{c+1}v_{c-1}\ldots v_0$. So, for example, the server $(c,v_{k-1}\ldots v_{c+1}*v_{c-1}\ldots v_0)$, where $*$ denotes that we may substitute in any number from $\{0,1,\ldots,\frac{n}{2}-1\}$, is connected to the switch $v_{k-1}\ldots v_{c+1}v_{c-1}\ldots v_0$ in column $c$, which in turn is connected to the server $(c+1,v_{k-1}\ldots v_{c+1}*v_{c-1}\ldots v_0)$. Similarly, the server $(c,v_{k-1}\ldots v_{c}*v_{c-2}\ldots v_0)$ is connected to the switch $v_{k-1}\ldots v_{c}v_{c-2}\ldots v_0$ in column $c-1$, which in turn is connected to the server $(c-1,v_{k-1}\ldots v_{c}*v_{c-2}\ldots v_0)$. The server-switch links for DPillar$_{6,3}$ can be visualized as in Fig.~\ref{DPillar}.

An alternative informal definition of DPillar$_{n,k}$ can be given. With reference to Fig.~\ref{DPillar}, we can replace every switch with a complete bipartite graph $K_{\frac{n}{2},\frac{n}{2}}$ (the bipartition is the obvious one). What results is the well-known \emph{wrapped butterfly network\/} (see, \emph{e.g.}, \cite{Lei92}; this network has been well-studied within the context of multiprocessor systems). The primary difference between DPillar$_{n,k}$ and the resulting wrapped butterfly network is that a switch in DPillar$_{n,k}$ enables direct server-to-server communication between servers connected to the same switch and in the same column, whereas such communication is absent in the wrapped butterfly network.

\subsection{Abstracting DPillar}\label{subsect:absDPillar}

We can abstract DPillar$_{n,k}$ as a digraph as follows: the nodes of this graph are the servers of DPillar$_{n,k}$; and there is an edge from a source-node to a target-node if there is a link from the corresponding source-server to a switch and a link from that switch to the corresponding target-server (so, the edges correspond to server-switch-server paths). There are $4$ types of edges in the digraph abstracting DPillar$_{n,k}$:

\begin{itemize}
\item[(\emph{i\/})] \emph{clockwise edges\/} (\emph{c-edges\/}) which are edges of the form 
  \begin{align*}
    ((c,v_{k-1}\ldots v_{c+1}v_c v_{c-1}\ldots v_0),
    (c+1,v_{k-1}\ldots v_{c+1}*v_{c-1}\ldots v_0))
  \end{align*}
\item[(\emph{ii\/})] \emph{anti-clockwise edges\/} (\emph{a-edges\/}) which are edges of the form 
  \begin{align*}
    ((c,v_{k-1}\ldots v_{c}v_{c-1} v_{c-2}\ldots v_0),   (c-1,v_{k-1}\ldots v_{c}*v_{c-2}\ldots v_0))
  \end{align*}
\item[(\emph{iii\/})] \emph{basic static edges\/} (\emph{b-edge\/}s) which are edges of the form 
  \begin{align*}
    ((c,v_{k-1}\ldots v_{c+1}v_{c} v_{c-1}\ldots v_0),   (c,v_{k-1}\ldots v_{c+1}*v_{c-1}\ldots v_0))
  \end{align*}
\item [(\emph{iv\/})] \emph{decremented static edges\/} (\emph{d-edges\/}) which are edges of the form
  \begin{align*}
    ((c,v_{k-1}\ldots v_{c}v_{c-1}v_{c-2}\ldots v_0),   (c,v_{k-1}\ldots v_{c}*v_{c-2}\ldots v_0)).
  \end{align*}
\end{itemize}
So, within our abstraction of DPillar$_{n,k}$ as a digraph, the nodes are the
servers and are located in columns $0,1,\ldots,k-1$ (as before) with all edges
joining nodes in consecutive columns (clockwise and anticlockwise edges) or
nodes in the same column (static edges). In fact, our digraph (where each node
has in- and out-degree $2n-2$) can also be thought of as an undirected graph
(that is regular of degree $2n-2$) as all edges come in oppositely oriented
pairs. Note that the clockwise (resp.\ anti-clockwise, basic static, decremented
static) edge above corresponds to a server-switch-server path in the DCN
DPillar$_{n,k}$ from a column $c$ server through a column $c$ (resp.\ $c-1$, $c$,
$c-1$) switch and on to a column $c+1$ (resp.\ $c-1$, $c$, $c$) server.
Henceforth, we denote the digraph abstracting DPillar$_{n,k}$ by DPillar$_{n,k}$
too (this causes no confusion). The abstraction of DPillar can be visualized as
in Fig.~\ref{DPillar} where we show how the switch $00$ in column $2$ gives rise
to a set of edges in the abstraction of DPillar as a graph. We annotate edges as
follows: an edge annotated `a' is an anti-clockwise edge relative to the node
$(0,000)$ (the arrow on the edge from $(0,000)$ denotes that the label is with
respect to $(0,000)$); an edge annotated `b' is a basic static edge relative to
node $(2,000)$; an edge annotated `c' is a clockwise edge relative to node
$(2,000)$; and an edge annotated `d' is a decremented static edge relative to
node $(0,000)$ (so, an edge has two labels: one relative to one incident node;
and another relative to the other incident node). In short, for some node, the
adjacent switch `to the right' gives rise to b-edges and c-edges, and the one
`to the left' gives rise to a-edges and d-edges.

\section{DPillar is a Cayley Graph}\label{sec:Cayley}

In this section, we prove that the digraph DPillar$_{n,k}$ is a Cayley graph,
and consequently node-symmetric (we exploit this node-symmetry later on in our
single-path routing algorithm and in our experimental work). Recall that a graph
is a \emph{Cayley graph} if the nodes can be labelled with the elements of a
(algebraic) group $G$ and there is a generating subset $S\subseteq G$ that is
closed under inverses so that every directed edge $(u,v)$ is labelled with an
element of $s\in S$ if, and only if, $us=v$ (within the group $G$). We say that
a digraph is \emph{node-symmetric} if given any $2$ distinct nodes $src$ and
$dst$, there is an automorphism (that is, a one-to-one mapping $\varphi$ of the
node-set onto itself such that if $(u,v)$ is an edge then
$(\varphi(u),\varphi(v))$ is an edge) mapping $src$ to $dst$. It is well-known,
and trivial to prove, that every Cayley graph is node-symmetric. The first paper
to establish that being a Cayley graph is a useful property for an
interconnection network is \cite{AK89} and since then, there has been much
research into representing interconnection networks using finite groups. Not
only do we immediately obtain that any Cayley graph is node-symmetric (which is
a fundamental property of interconnection networks \cite{DT04}) but Cayley
graphs have been shown to be relevant to various networks in a variety of ways;
for example, with regard to the design of interconnection networks by pruning
nodes and edges from tori \cite{XB07}, the design of wireless DCNs \cite{SSW13},
and the design of high-dimensional mesh-based interconnection networks
\cite{CMB15}.

\subsection{DPillar Symmetry}\label{subsect:symDPillar}

Whilst it was stated in \cite{LYY12} that the DCN DPillar is `symmetric', it was not stated as to what `symmetric' meant (hence, there was no proof of `symmetry'). Our  main intention is to show that DPillar is node-symmetric (defined above) but we do this by proving that DPillar is a Cayley graph.

\begin{lemma}\label{lem:CayGra}
The digraph DPillar$_{n,k}$ is a Cayley graph.
\end{lemma}

\begin{proof}
  Our proof is related to the proof in \cite{FC98} that the wrapped butterfly
  network (called the \emph{cyclic cube} in \cite{FC98}) is a Cayley graph.
  The full proof can be found in the supplemental material.
\end{proof}

We obtain the immediate corollary.

\begin{corollary}\label{lem:sym}
The digraph DPillar$_{n,k}$ is node-symmetric.
\end{corollary}

\section{Abstracting routing in DPillar}\label{sec:routing}

In this section, we abstract the problem of finding a path in the digraph DPillar$_{n,k}$ from a given source-node to a given destination-node so that ultimately this problem is equivalent to finding a path from a source-node to a destination-node in a cycle of length $k$ but where the actual node-to-node moves are more complicated than in a digraph. We also explain the single-path routing algorithm from \cite{LYY12}.

\subsection{Fixing bits}

It is important to appreciate what might be accomplished by moving along one of
the $4$ different types of edge highlighted above. Suppose that we are
attempting to move from some source-node $src$ to some destination-node $dst$
within DPillar$_{n,k}$ and that we are currently at some node in column $c$. We
can choose a clockwise (resp.\ anti-clockwise, basic static, decremented static)
edge so as to set the $c$th (resp.\ $(c-1)$th, $c$th, $(c-1)$th) bit in the
row-index to whatever value from $\{0,1,\ldots,\frac{n}{2}-1\}$ that we like.
Consequently, by choosing a clockwise (resp.\ anti-clockwise, basic static,
decremented static) edge along which to move, we can `fix' the $c$th (resp.
$(c-1)$th, $c$th, $(c-1)$th) bit of the row-index so that it matches that of the
destination-node. We say that: a clockwise edge \emph{covers\/} the column in
which its source-node lies; an anti-clockwise edge \emph{covers\/} the column in
which its target-node lies; a basic static edge \emph{covers\/} the column in
which both its source- and target-nodes lie; and a decremented static edge
\emph{covers\/} the column that is adjacent in an anti-clockwise direction to
the column in which both its source- and target-nodes lie. Thus, if we wish to
move along some path from $src$ to $dst$ then we need to ensure that we move
from column to column so as to fix all of the bits of the row-index that need
fixing, but so that we don't subsequently `unfix' them, \emph{and\/} so that we
end up in the column within which $dst$ resides (with regard to not `unfixing' a
bit, note that we can always move from a node in one column to a node in an
adjacent column so that the row-index remains unchanged). This is equivalent to
moving from column to column so that every row-index bit-position, \emph{i.e.},
column, where the bit values of $src$ and $dst$ differ is necessarily covered by
some edge and so that we end up in the column within which $dst$ resides. If we
are looking for a shortest path from $src$ to $dst$ then we have to do this
using as few moves as possible. Of course, any path of length $l$ in our
abstraction of DPillar$_{n,k}$ as a digraph translates to a path consisting of
$l$ server-switch-server link-pairs in the DCN DPillar$_{n,k}$, and \emph{vice
  versa\/} (for the sake of uniformity, we measure the length of
server-to-server paths in the DCN DPillar in terms of the number of
server-switch-server link-pairs in the path; this is also common practice in the
DCN community).

As an illustration, suppose we are at $(1,12530)$ in DPillar$_{6,5}$ and wish to get to the destination $(4,54314)$. If $x$ denotes any element of $\{0,1,2,3,4,5\}$, there is: an anti-clockwise edge taking us to $(0,1253x)$; a basic static edge taking us to $(1,125x0)$; a clockwise edge taking us to $(2,125x0)$; and a decremented static edge taking us to $(1,1253x)$. Given our destination, when we move we can choose $x$ accordingly and fix the appropriate bit so that we move: via an anti-clockwise edge to $(0,1253\underline{4})$; via a basic static edge to $(1,125\underline{1} 0)$; via a clockwise edge to $(2,125\underline{1}0)$; or via a decremented static edge to $(1,1253\underline{4} )$.

\subsection{Another abstraction}

A crucial observation arising from the above discussion is that when routing in DPillar$_{n,k}$, the actual value of some bit in a row-index of some node is unimportant: what matters is whether this value is equal to or different from the value of the corresponding bit in the row-index of the destination-node (that is, whether the bit needs to be `fixed' or not). Consequently, in order to solve the problem of finding a path from $src$, which lies in column $src^\prime$, to $dst$, which lies in column $dst^\prime$, in DPillar$_{n,k}$, we can abstract the problem as a (more involved) routing problem in the following digraph $G_{n,k}(src^\prime,dst^\prime)$:
\begin{itemize}
\item we think of there being one node for each of the $k$ columns of nodes of DPillar$_{n,k}$ with nodes in $G_{n,k}(src^\prime,dst^\prime)$ that correspond to adjacent columns being joined by an oppositely oriented pair of edges (so, we can also think of $G_{n,k}(src^\prime,dst^\prime)$ as an undirected cycle of length $k$)
\item we \emph{mark\/} every node $c$, corresponding to some column $c$ (or, alternatively, some bit-position $c$ in the row-index of some node of DPillar$_{n,k}$) that needs to be covered (because bit $c$ of the row-index of $src$ is different from bit $c$ of the row-index of $dst$), with the set of marked nodes being denoted by $B$
\item we move from node to node in $G_{n,k}(src^\prime,dst^\prime)$, starting at the node $src^\prime$ so as to end at the node $dst^\prime$ and making \emph{moves\/} where:
\begin{itemize}
\item[(\emph{i\/})] a \emph{c-move\/} means we move from node $c$ to node $c+1$ and such a move \emph{covers\/} node $c$
\item[(\emph{ii\/})] an \emph{a-move\/} means we move from node $c$ to node $c-1$ and such a move \emph{covers\/} node $c-1$
\item[(\emph{iii\/})] a \emph{b-move\/} means we stay at node $c$ and such a move \emph{covers\/} node $c$
\item[(\emph{iv\/})] a \emph{d-move\/} means we stay at node $c$ and such a move \emph{covers\/} node $c-1$
\end{itemize}
\end{itemize}
(note the correspondence between the above moves and the edge types given in Section~\ref{subsect:absDPillar}). 
We call $G_{n,k}(src^\prime,dst^\prime)$ a \emph{marked cycle\/}. Note that it might be the case that $src^\prime=dst^\prime$ in $G_{n,k}(src^\prime,dst^\prime)$ (this would mean that the nodes $src$ and $dst$ lie in the same column in DPillar$_{n,k}$).

With regard to our illustration in the previous section, the edge from $(1,12530)$: to $(0,12534)$ results in an a-move covering node $0$ in the marked cycle; to $(1,12510)$ results in a b-move covering node $1$ in the marked cycle; to $(2,12510)$  results in a c-move covering node $1$ in the marked cycle; and to $(1,12534)$ results in a d-move covering node $0$ in the marked cycle.

It should be clear as to how moves in the marked cycle $G_{n,k}(src^\prime,dst^\prime)$ correspond to moves along corresponding edges in DPillar$_{n,k}$ (and so to server-switch-server link-pairs in the DCN DPillar$_{n,k}$) with the coverage of a node in $G_{n,k}(src^\prime,dst^\prime)$ and a node of DPillar$_{n,k}$ being in direct correspondence. A \emph{path\/} in $G_{n,k}(src^\prime,dst^\prime)$ is a sequence of moves leading from $src^\prime$ to $dst^\prime$ and corresponds to a path in DPillar$_{n,k}$ from node $src$ to node $dst$ (and \emph{vice versa\/}) with the lengths of the two paths being identical. Consequently, in order to find a shortest path from $src$ to $dst$ in the DCN DPillar$_{n,k}$, it suffices to find a shortest path in the marked cycle $G_{n,k}(src^\prime,dst^\prime)$ (from the node $src^\prime$ to the node $dst^\prime$) so that every marked node is covered by a move. Note that if $src^\prime=dst^\prime$ then the empty sequence of moves does not constitute a legitimate path.

\subsection{Basic routing in DPillar}

Before we continue, let us discuss the single-path routing algorithm for DPillar as detailed in \cite{LYY12}; we refer to this algorithm as DPillarSP. The routing algorithm DPillarSP operates in $2$ phases: in the first phase (the so-called `helix' phase), a path in the DCN DPillar$_{n,k}$ is chosen so that movement is always in a clockwise direction (that is, the column-index is always incremented) or always in an anti-clockwise direction (that is, the column-index is always decremented) in order that the row-index is `fixed' so that it is identical to that of the destination-node; and in the second phase (the so-called `ring' phase), a path is subsequently chosen so as to reach the destination-node without amending the row-index and so that movement is in the same direction as in the first phase. Although not explicitly mentioned when discussing their algorithm, it is clear that the time complexity of the single-path routing algorithm from \cite{LYY12} is $O(k)$ (we have suppressed the $\log n$ component required to represent each bit-value).

It is stated in \cite[Section 3.1]{LYY12} that this single-direction movement is so that `loops' might be avoided. While this statement was not explained further, it is probable that what was meant by `loops' was a loop within a single route for a source-destination pair. Of course, our shortest-path routing algorithm means that loops in a single path will never occur. Alternatively (though unlikely), the rationale for the decision in \cite{LYY12} to restrict to single-direction movement might have been to avoid either network-level deadlock or livelock due to dependency loops (see, \emph{e.g.}, \cite[Ch. 14]{DT04}). Irrespective of the intentions in \cite{LYY12}, it is worth commenting on the potential for deadlocks in DPillar and server-centric DCNs in general. Given that the topology of DPillar is basically a sophisticated ring of columns, moving in a single direction does not completely prevent dependency loops from appearing. We give an example in Fig.~\ref{deadlock} where there is a (bold) route from $(0,000)$ to $(2,200)$ and a (dotted) route from $(1,200)$ to $(1,000)$ so that there is a cyclic dependency graph, due to the shared switches $(0,00)$ and $(1,20)$, even though we are using single-direction routing.  Nevertheless, there are many reasons to believe that, in the context of server-centric DCNs based on COTS hardware and software (\emph{i.e.}, Ethernet hardware and TCP/IP stack), network level deadlocks should be a minor concern. First, commodity Ethernet hardware uses packet-switching which prevents network frames from spreading across many network components; therefore a cyclic dependency between frames is unlikely to happen. Second, servers have virtually unlimited memory (and indeed, many orders of magnitude more than switches); hence we can assume infinite FIFOs at the servers. Considering that one of the necessary conditions for deadlocks to appear is for FIFOs to become full, it is, again, very unlikely that we end up in a deadlock situation. Finally, in the very unlikely situation of a cyclic dependency appearing and all the FIFOs becoming full, the packet-dropping mechanism of Ethernet-based hardware provides seamless deadlock recovery, whereas TCP ensures data delivery. The upshot is that deadlocks are not a primary concern in DCNs.

\begin{figure}[htb]
	\centering
		\includegraphics[width=.8\linewidth]{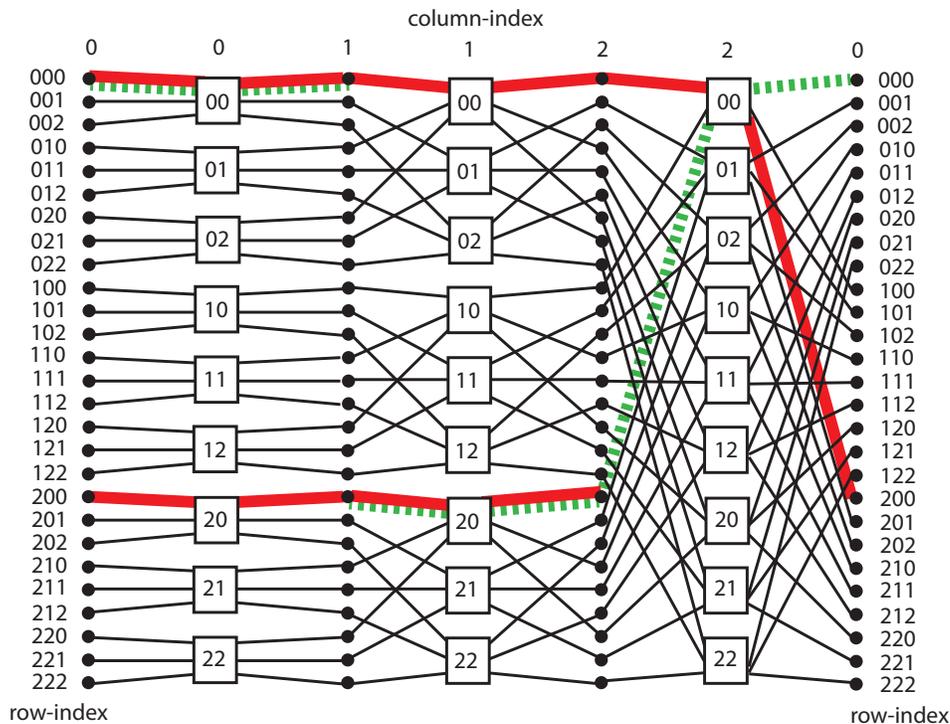}
	\caption{A dependency loop between two routes with DPillarSP. Server (0, 000) sends to (2, 200) and server (1, 200) sends to (1, 000). The paths through switches (0, 00) and (1, 20) are conflicted.}
	\label{deadlock}
\end{figure}

It is very easy to see (by looking at some typical source-destination examples) that the routing algorithm DPillarSP is by no means optimal and that more often than not much shorter paths exist (an upper bound of $2k - 1$ on the lengths of paths produced was stated in \cite{LYY12}). For example, if one chooses to route in a clockwise fashion in DPillar$_{n,k}$ with the source $(0,00\ldots 0)$ and the destination $(1,10\ldots 0)$ then the DPillarSP yields a path of length $k + 1$, and if one routes in an anti-clockwise fashion then the algorithm also yields a path of length $k - 1$; however, a shortest path has length $2$ (a d-move followed by a c-move). Our contention is that by relaxing this insistence on single-direction movement, we can obtain a much improved routing algorithm; indeed, as we shall see, we develop an optimal single-path routing algorithm (where the implementation overheads are negligible and where there are significant practical benefits).

\section{Routing in a marked cycle}\label{sec:routingmarked}

We begin by making some initial observations as regards routing along a shortest path (from $src^\prime$ to $dst^\prime$) in a marked cycle $G_{n,k}(src^\prime,dst^\prime)$ before proving that any such shortest path has a restricted structure.

\subsection{Some initial observations}

Henceforth, $\rho$ is a shortest path from $src^\prime$ to $dst^\prime$ in $G_{n,k}(src^\prime,dst^\prime)$. Consider two consecutive moves in $\rho$. We can often rule out consecutive pairs of moves. For example, suppose that we have within $\rho$ a c-move followed by an a-move. We can replace this pair within $\rho$ by a b-move so as to obtain a path with identical coverage to $\rho$ and which is shorter. This yields a contradiction. Similarly, suppose that we have an a-move followed by a c-move within $\rho$. We can replace this pair within $\rho$ by a d-move so as to again obtain a contradiction. In Table~\ref{tablemoves}, we detail all pairs of consecutive moves in $\rho$ that are forbidden by including the substitution that would result in a shorter path that has equivalent coverage. In this table, the first move is detailed in the rows and the second move in the columns. A blank cell means that the corresponding pair of moves cannot immediately be ruled out.

\begin{table}[ht]
\caption{Disallowed pairs of moves.}
\centering
\begin{tabular}{|l | c | c | c | c|}
\hline
& a-move & b-move & c-move & d-move\\
\hline
a-move &  & a-move & d-move &  \\
b-move &  & b-move & c-move &  \\
c-move & b-move &  &  & c-move \\
d-move & a-move &  &  & d-move \\
\hline
\end{tabular}
\label{tablemoves}
\end{table}

For clarity, rather than say, for example, `a c-move followed by an a-move', in future we will simply write $ca$ to denote this circumstance. Consequently, subsequences of moves within $\rho$ will be written as strings over $\{a,b,c,d\}$ (as will $\rho$ itself) and we compress subsequences of the same symbol, such as $aaaa$, by using powers, such as $a^4$.

We can say more. If we have a subsequence of moves $bd$ then this has the same effect as the subsequence $db$, and so we may suppose that a subsequence $db$ within $\rho$ is forbidden. Also, note that if $\rho$ has length at least $3$ then we cannot have a subsequence $bd$: 
\begin{itemize}
\item a subsequence $bdb$ can be replaced by $bd$; a subsequence $bdc$ can be replaced by $dc$; and we cannot have a subsequence $da$ or $dd$
\item a subsequence $cbd$ can be replaced by $cb$; a subsequence $dbd$ can be replaced by $db$; and we cannot have a subsequence $ab$ or $bb$.
\end{itemize}
Consequently, if $\rho$ has length at least $3$ then:
\begin{itemize}
\item if a c-move is not the final move of $\rho$ then it must be followed by another c-move or a b-move
\item if an a-move is not the final move of $\rho$ then it must be followed by another a-move or a d-move
\item if a b-move is not the final (resp.\ first) move of $\rho$ then it must be followed by an a-move (resp.\ preceded by a c-move)
\item if a d-move is not the final (resp.\ first) move of $\rho$ then it must be followed by a c-move (resp.\ preceded by an a-move).
\end{itemize}
Consequently, if $\rho$ has length at least $3$ then it must be of one of two forms: \begin{itemize}
\item[(1)] possibly a d-move (but maybe not) followed by a sequence of c-moves followed by a b-move followed by a sequence of a-moves followed by a d-move followed by a sequence of c-moves followed by $\dots$ followed by a sequence of c-moves (resp.\ a-moves) possibly followed by a b-move (resp.\ d-move); that is, $$d^\epsilon c^{i_1}ba^{j_1}dc^{i_2}\ldots c^{i_m}b^\delta\mbox{ or }d^\epsilon c^{i_1}ba^{j_1}dc^{i_2}\ldots a^{j_m}d^\delta,$$ for some $m\geq 1$, where $i_1,i_2,\ldots,i_m,j_1,j_2,\ldots,j_m \geq 1$ and where $\epsilon,\delta\in\{0,1\}$
\item[(2)] possibly a b-move followed by a sequence of a-moves followed by a d-move followed by a sequence of c-moves followed by a b-move followed by a sequence of a-moves followed by $\dots$ followed by a sequence of a-moves (resp.\ c-moves) possibly followed by a d-move (resp.\ b-move); that is, $$b^\epsilon a^{i_1}dc^{j_1}ba^{i_2}\ldots a^{i_m}d^\delta\mbox{ or }b^\epsilon a^{i_1}dc^{j_1}ba^{i_2}\ldots c^{j_m}b^\delta,$$ for some $m\geq 1$, where $i_1,i_2,\ldots,i_m,j_1,j_2,\ldots,j_m \geq 1$ and where $\epsilon,\delta\in\{0,1\}$
\end{itemize}
(when we say `sequence', above, we mean `non-empty sequence').

\subsection{Restricting the number of turns} \label{sec:turns}

If we have a subsequence $cba$ in $\rho$ then we say that an \emph{anti-clockwise turn}, or simply an \emph{a-turn}, occurs at the b-move; similarly, if we have a subsequence $adc$ then we say that a \emph{clockwise turn}, or simply a \emph{c-turn}, occurs at the d-move. Note that if we have an a-turn in $\rho$ then the node at which this turn occurs, \emph{i.e.}, the node that is covered by the d-move, must be marked in $G_{n,k}(src^\prime,dst^\prime)$ as otherwise we could delete the corresponding d-move from $\rho$ and still have a sequence from $src^\prime$ to $dst^\prime$ covering all the marked nodes, which would yield a contradiction. Similarly, if we have a c-turn then the node at which this c-turn occurs, \emph{i.e.}, the node that is covered by the b-move, must be marked. We will use these observations later; but now we prove that any shortest path $\rho$ must contain at most $2$ turns. 

Suppose that $\rho$ is a shortest path and has at least $3$ turns. What we do now is undertake a case by case analysis of the different configurations that might arise. These cases arise from the forms derived at the end of the previous subsection: the first two cases correspond to form (1) and the next two cases to form (2). The technique employed in each case is to modify the path $\rho$, by replacing sequences of moves within $\rho$, so as to obtain a new path that has the same coverage but is shorter; this yields a contradiction to our assumption that $\rho$ has at least $3$ turns.\smallskip

\noindent\underline{Case (\emph{a\/})}: Suppose that $\rho$ is of form (1) and has a prefix $\rho^\prime$ of the form $c^iba^jdc^lba$, where $i,j,l\geq 1$.\smallskip

\noindent By this we mean that $\rho$ begins with $i$ c-moves followed by a b-move followed by $j$ a-moves followed by a d-move followed by $l$ c-moves followed by a b-move followed by an a-move. 

If $j < i$ then we can replace the prefix $c^iba^jdc$ in $\rho^\prime$ with $c^iba^{j-1}$ and still obtain the same coverage; this contradicts that $\rho$ is a shortest path (note that we have actually only assumed so far that $\rho$ has $2$ turns). If $j=i$ then we can replace the prefix $c^iba^idc$ in $\rho^\prime$ with $dc^iba^{i-1}$ so as to obtain a contradiction (we have still actually only assumed that $\rho$ has $2$ turns). Hence, we must have that $j> i$. Suppose that $j\geq l > j-i$. We can replace the prefix $c^iba^jdc^l$ in $\rho^\prime$ with $a^{j-i}dc^jba^{j-l}$ so as to obtain a contradiction (we have still actually only assumed that $\rho$ has $2$ turns). Hence, $j>i$ and either $l \leq j-i$ or $l > j$. 

Suppose that $l > j$. We can replace the prefix $c^iba^jdc^l$ in $\rho^\prime$ with $a^{j-i}dc^l$ so as to obtain a contradiction (we have still actually only assumed that $\rho$ has $2$ turns). Hence, we must have that $j >i$ and $l \leq j-i$. However, if we replace $\rho^\prime$ with $c^iba^jdc^{l-1}$ then we obtain a contradiction (here we do use the fact that $\rho$ has at least $3$ turns). So, $\rho$ has at most $2$ turns and if it has $2$ turns then $\rho$ is of the form $c^iba^jdc^l$ where $j >i$ and $l \leq j-i$.

We can say more if $\rho$ has $2$ turns. Suppose that $j\geq k-1$. The b-move can be deleted from $\rho^\prime$ and we obtain a contradiction. Hence, if $\rho$ has $2$ turns then $\rho$ is of the form $c^iba^jdc^l$ where $k-1 > j >i\geq 1$ and $1\leq l \leq j-i$. We can visualize $\rho$ as in Fig.~\ref{2turns}(\emph{i\/}). The marked cycle $G_{n,k}(src^\prime,dst^\prime)$ is shown as a cycle where a black node denotes a node of $B$; that is, a node that needs to be covered by some path in $G_{n,k}(src^\prime,dst^\prime)$ (from $src^\prime$ to $dst^\prime$, with $0=src^\prime\neq dst^\prime=x$ in this illustration). The path $\rho$ is depicted as a dotted line partitioned into composite moves.\smallskip

\begin{figure}[t]
\centering
\includegraphics[width=0.8\linewidth]{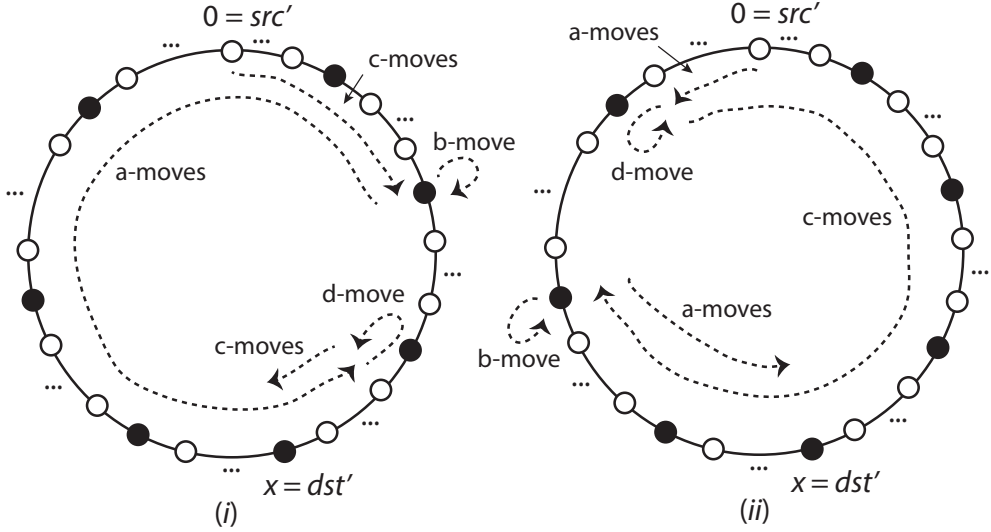}
\caption{Visualizing paths with $2$ turns.}\label{2turns}
\end{figure}

\noindent\underline{Case (\emph{b\/})}: Suppose that $\rho$ is of form (1) and has a prefix $\rho^\prime$ of the form $dc^iba^jdc^lba$, where $i,j,l\geq 1$.\smallskip

\noindent If $j \leq i$ then we can replace the prefix $dc^iba^jdc$ in $\rho^\prime$ with $dc^iba^jc$ so as to obtain a contradiction, and if $j > i$ then we can delete the first d-move from $\rho$ to obtain a contradiction. Hence, if $\rho$ starts with a d-move then it has at most $1$ turn.\smallskip

\noindent\underline{Case (\emph{c\/})}: Suppose that $\rho$ is of form (2) and has a prefix $\rho^\prime$ of the form $a^idc^jba^ldc$, where $i,j,l\geq 1$.\smallskip

\noindent If $j< i$ then we can replace the prefix $a^idc^jba$ in $\rho^\prime$ with $a^idc^{j-1}$ so as to obtain a contradiction. If $i=j$ then we can replace the prefix $a^idc^iba$ in $\rho^\prime$ with $ba^idc^{i-1}$ so as to obtain a contradiction. Hence, $j > i$. 

Suppose that $j \geq l > j-i$. We can replace the prefix $a^idc^jba^l$ in $\rho$ with $c^{j-i}ba^jdc^{j-l}$ so as to obtain a contradiction. Suppose that $l > j$. We can delete the first occurrence of a d-move in $\rho$ so as to obtain a contradiction. Hence, $l \leq j-i$. Note that if $\rho$ has $2$ turns then $\rho$ is of the form $a^idc^jba^l$ where $j > i$ and $l \leq j-i$. Alternatively, suppose that $\rho$ has at least $3$ turns. We can replace the prefix $a^idc^jba^ldc$ in $\rho$ with $a^idc^jbc^{l-1}$ so as to obtain a contradiction. Hence, $\rho$ has at most $2$ turns.

We can say more if $\rho$ has $2$ turns. Suppose that $j\geq k-1$. The d-move can be deleted from $\rho^\prime$ and we obtain a contradiction. Hence, if $\rho$ has $2$ turns then $\rho$ is of the form $a^idc^jbd^l$ where $k-1>j >i\geq 1$ and $1 \leq l \leq j-i$. We can visualize $\rho$ as in Fig.~\ref{2turns}(\emph{ii\/}).\smallskip

\noindent\underline{Case (\emph{d\/})}: Suppose that $\rho$ is of form (2) and has a prefix $\rho^\prime$ of the form $ba^idc^jba^ldc$, where $i,j,l\geq 1$.\smallskip

\noindent If $j\leq i$ then we can replace the prefix $ba^idc^jba$ with $ba^idc^ja$ so as to obtain a contradiction, and if $j > i$ then we can delete the first b-move from $\rho$ to obtain a contradiction. Hence, if $\rho$ starts with a b-move then it has at most $1$ turn.\smallskip

So, we have proven the following lemma.

\begin{lemma}\label{lem:turns}If $\rho$ is a shortest path (from $src^\prime$ to $dst^\prime$) in $G_{n,k}(src^\prime,dst^\prime)$ then $\rho$ has at most $2$ turns, and if $\rho$ has $2$ turns then it must be of the form $c^iba^jdc^l$ or $a^idc^jba^l$, where $k-1 > j > i \geq 1$ and $1 \leq l\leq j-i$.
\end{lemma}

With reference to Fig.~\ref{2turns}, the numerical constraints in Lemma~\ref{lem:turns} mean that there is no interaction or overlap involving the $2$ turns in $\rho$.

\section{An optimal routing algorithm for DPillar}\label{sec:algorithm}

We now develop an optimal single-path routing algorithm for DPillar, based around Lemma~\ref{lem:turns}. We do this by finding a small set $\Pi$ of paths (from $src^\prime$ to $dst^\prime$) in $G_{n,k}(src^\prime,dst^\prime)$ so that at least one of these paths is a shortest path (and consequently we obtain a shortest path in the DCN DPillar$_{n,k}$). By Lemma~\ref{lem:sym}, we may assume that $src=(0,00\ldots0)$ and $dst=(x,v_{k-1}v_{k-2}\ldots v_0)$, and by Lemma~\ref{lem:turns}, we may assume that any shortest path has at most $2$ turns.

Our technique is as follows. Essentially, we want to make the set $\Pi$ as small as possible; that is, we want our resulting algorithm to have to consider as few paths as possible (when looking for the shortest). Lemma~\ref{lem:turns} precisely describes the set of paths we need to consider from the paths involving exactly $2$ turns; of course, we also need to consider paths involving $1$ or $0$ turns (if they exist). There are different situations depending upon the distribution of the marked nodes needing to be covered; in particular, upon the distribution of marked nodes along the natural clockwise and anti-clockwise paths from the source to the destination on the marked cycle, assuming the source and destination to be distinct (this is the case in Section~\ref{distinctsourcedest}; the case when the source and destination are the same is considered in Section~\ref{samesourcedest}). Sometimes the distribution of marked nodes rules out the possibility of certain types of paths.

\subsection{Building our set of paths when $x \neq 0$}\label{distinctsourcedest}

We first suppose that $0\neq x$. Let $B=\{i: 0\leq i\leq k-1, v_i\neq 0\}$ (that
is, the bit-positions that need to be `fixed'). Suppose that $B\setminus\{0,x\}
= \{i_l: 1 \leq l \leq r\}\cup\{j_l: 1 \leq l \leq s\}$ so that we have $0 < j_s
< j_{s-1} < \ldots < j_1 < C < i_1 < i_2 < \ldots < i_r < k$ (we might have that
either $r$ or $s$ is $0$, when the corresponding set is empty). If $r \geq 2$
then define $\delta_l = i_{l+1}-i_l$, for $l=1,2,\ldots,r-1$, with
$\delta=\max\{\delta_l:l=1,2,\ldots,r-1\}$; and if $s\geq 2$ then define
$\epsilon_l = j_l-j_{l+1}$, for $l=1,2,\ldots,s-1$, with
$\epsilon=\max\{\epsilon_l:l=1,2,\ldots,s-1\}$. Also: define $\Delta_0 = 1$
(resp.\ $0$), if $0\in B$ (resp.\ $0\not\in B$); and $\Delta_x = 1$ (resp.\
$0$), if $x\in B$ (resp.\ $x\not\in B$). We can visualize the resulting marked
cycle $G_{n,k}(0,x)$ as in Fig.~\ref{setup}(\emph{i\/}). Note that in this
particular illustration $0\not \in B$ and $x\in B$; so, $\Delta_0=0$ and
$\Delta_x=1$. Of course, what we are looking for is a sequence of (a-, b-, c-
and d-)moves that will take us from $0$ to $x$ in $G_{n,k}(0,x)$ so that all
nodes of $B$ have been covered.

\begin{figure}[t]
\centering
\includegraphics[width=.8\linewidth]{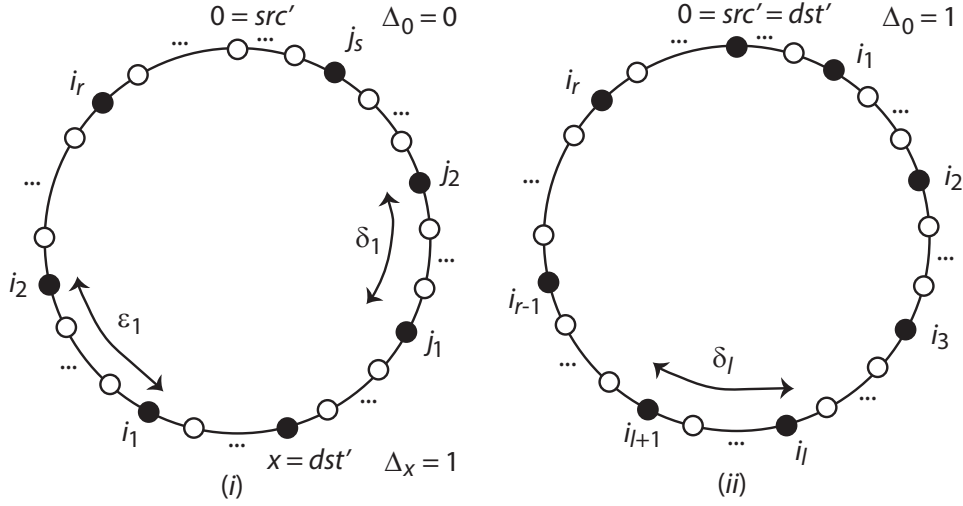}
\caption{Visualizing our notation.}\label{setup}
\end{figure}

In what follows, we examine different scenarios involving the number of marked nodes, $r$, and also the number of marked nodes, $s$. Each scenario for $r$ contributes certain paths to $\Pi$ as does each scenario for $s$. Note that perhaps the most obvious paths to consider as potential members of $\Pi$ are the paths $c^{k+x}$ and $a^{2k-x}$ which have lengths $k+x$ and $2k-x$, respectively. So, we begin by setting $\Pi=\{c^{k+x},a^{2k-x}\}$.

From Lemma~\ref{lem:turns}, any shortest path $\rho$ from $0$ to $x$ having $2$ turns requires that $r\geq 2$ or $s\geq 2$ and that both nodes at which these turns occur are different from $0$ and $x$ and lie on the anti-clockwise path from $0$ to $x$ or on the clockwise path from $0$ to $x$, accordingly. Recall also that the node at which any turn occurs on a shortest path $\rho$ is necessarily a marked node (irrespective of the number of turns in $\rho$).\smallskip

\noindent\underline{Case (\emph{a\/})}: Suppose that $r=0$.\smallskip

\noindent In this scenario, we contribute either the path $c^xb$ to $\Pi$, if $x\in B$, or the path $c^x$ to $\Pi$, if $x\not\in B$; either way, the length of the path contributed is $x+\Delta_x$.\smallskip

\noindent\underline{Case (\emph{b\/})}: Suppose that $s=0$.\smallskip

\noindent In this scenario, we contribute either the path $ba^{k-x}$ to $\Pi$, if $0\in B$, or the path $a^{k-x}$ to $\Pi$, if $0\not\in B$; either way, the length of the path contributed is $k-x+\Delta_0$.\smallskip

\noindent\underline{Case (\emph{c\/})}: Suppose that $r=1$.\smallskip

\noindent In this scenario, we contribute $2$ paths to $\Pi$. If $x\in B$ then we contribute the path $a^{k-i_1-1}dc^{k-i_1-1+x}b$ to $\Pi$, or if $x\not\in B$ then we contribute the path $a^{k-i_1-1}dc^{k-i_1-1+x}$ to $\Pi$; either way, the length of the resulting path is $ 2k-2i_1+x-1+\Delta_x$. We also contribute the path $c^{i_1}ba^{i_1-x}$ to $\Pi$ of length $2i_1-x+1$. There is potentially another path when $i_1=x+1$ and $x\in B$, namely $a^{k-x-1}dc^{k-1}$, but the length of this path is $2k-x-1$ which is greater than $2k-x-3+\Delta_x$, which in turn is $ 2k-2i_1+x-1+\Delta_x$ evaluated with $i_1=x+1$.\smallskip

\noindent\underline{Case (\emph{d\/})}: Suppose that $s=1$.\smallskip

\noindent In this scenario, we contribute $2$ paths to $\Pi$. If $0\in B$ then we contribute the path $ba^{k-j_1-1}dc^{x-j_1-1}$ to $\Pi$, or if $0\not\in B$ then we contribute the path $a^{k-j_1-1}dc^{x-j_1-1}$ to $\Pi$; either way, the length of the resulting path is $k-2j_1+x-1+\Delta_0$. We also contribute the path $c^{j_1}ba^{k+j_1-x}$ to $\Pi$ of length $k+2j_1-x+1$. There is potentially another path when $j_1=1$ and $0\in B$, namely $a^{k-1}dc^{x-1}$, but the length of this path is $k+x-1$ which is greater than $k+x-3+\Delta_0$, which in turn is $k-2j_1+x-1+\Delta_0$ evaluated with $j_1=1$.\smallskip

\noindent\underline{Case (\emph{e\/})}: Suppose that $r\geq 2$.\smallskip

\noindent In this scenario, we contribute $r+1$ paths to $\Pi$. For each $l\in\{1,2,\ldots,r-1\}$, we contribute the path $a^{k-i_{l+1}-1} dc^{k-i_{l+1}-1+i_l}ba^{i_l-x}$ to $\Pi$ of length $2k-2\delta_l-x$. If $x\in B$ then we contribute the path $a^{k-i_1-1}dc^{k-i_1-1+x}b$ to $\Pi$, or if $x\not\in B$ then we contribute the path $a^{k-i_1-1}dc^{k-i_1-1+x}$ to $\Pi$; either way, the length of the path is $2k-2i_1+x-1+\Delta_x$. We also contribute the path $c^{i_r}ba^{i_r-x}$ to $\Pi$ of length $2i_r-x+1$. (These last $2$ paths mirror those constructed in Case (\emph{c\/}).)\smallskip

\noindent\underline{Case (\emph{f\/})}: Suppose that $s\geq 2$.\smallskip

\noindent In this scenario, we contribute $s+1$ paths to $\Pi$. For each $l\in\{1,2,\ldots,s-1\}$, we contribute the path $c^{j_{l+1}} ba^{j_{l+1}+k-j_l-1}da^{x-j_l-1}$ to $\Pi$ of length $k-2\epsilon_l+x$. If $0\in B$ then we contribute the path $ba^{k-j_s-1}dc^{x-j_s-1}$ to $\Pi$, or if $0\not\in B$ then we contribute the path $a^{k-j_s-1}dc^{x-j_s-1}$ to $\Pi$; either way, the length of the path is $k-2j_s+x-1+\Delta_0$. We also contribute the path $c^{j_1}ba^{j_1+k-x}$ to $\Pi$ of length $k+2j_1-x+1$. (These last $2$ paths mirror those constructed in Case (\emph{c\/}).)\smallskip

Thus, our set $\Pi$ of potential shortest paths contains $r+s+2$ paths (from which at least one is a shortest path).

\subsection{Building our set of paths when $x = 0$}\label{samesourcedest}

Now we suppose that $x=0$. We proceed as we did above and build a set $\Pi$ of potential shortest paths. Let $B=\{i: 0\leq i\leq k-1, v_i\neq 0\}$. Suppose that $B\setminus\{0\} = \{i_l: 1 \leq l \leq r\}$ so that we have $0 < i_1 < i_2 < \ldots < i_r < k$ (we might have that $r$ is $0$ when the corresponding set is empty). If $r \geq 2$ then define $\delta_l = i_{l+1}-i_l$, for $l=1,2,\ldots,r-1$, with $\delta=\max\{\delta_l:l=1,2,\ldots,r-1\}$. We define $\Delta_0=1$, if $0\in B$, and $\Delta_0=0$, if $0\not\in B$. We can visualize the resulting marked cycle $G_{n,k}(0,0)$ as in Fig.~\ref{setup}(\emph{ii\/}). Again, the most obvious path to consider is $c^k$ (or $a^k$) which has length $k$. We begin by setting $\Pi=\{c^k\}$. \smallskip

\noindent\underline{Case(\emph{a\/})}: Suppose that $r=0$.\smallskip

\noindent In this scenario, we contribute the path $b$ of length $1$ (note that in this case the node $0$ is necessarily marked as we originally assumed that we started with distinct source and destination servers in the DCN DPillar$_{n,k}$).\smallskip

\noindent\underline{Case(\emph{b\/})}: Suppose that $r=1$.\smallskip

\noindent If $i_1=k-1$ then we contribute the path $bd$, if $0\in B$, and the path $d$, if $0\not\in B$; either way, the path has length $1+\Delta_0$. If $1=i_1\neq k-1$ then we contribute the path $cba$ of length $3$. If $1\neq i_1 \neq k-1$ then we contribute $2$ paths. The first of these paths is the path $ba^{k-i_1-1}dc^{k-i_1-1}$, if $0\in B$, and the path $a^{k-i_1-1}dc^{k-i_1-1}$, if $0\not\in B$; either way, this path has length $2k-2i_1-1+\Delta_0$. The second of these paths is the path $c^{i_1}ba^{i_1}$ of length $2i_1+1$.\smallskip

\noindent\underline{Case(\emph{c\/})}: Suppose that $r\geq 2$.\smallskip

\noindent In this scenario, we contribute $r+1$ paths to $\Pi$. For each $l\in\{1,2,\ldots,r-1\}$, we contribute the path $a^{k-i_{l+1}-1} dc^{k-i_{l+1}-1+i_l}ba^{i_l}$ to $\Pi$ of length $2k-2\delta_l$. If $0\in B$ then we contribute the path $ba^{k-i_1-1}dc^{k-i_1-1}$ to $\Pi$, or if $0\not\in B$ then we contribute the path $a^{k-i_1-1}dc^{k-i_1-1}$ to $\Pi$; either way, this path has length $2k-2i_1-1+\Delta_0$. We also contribute the path $c^{i_r}ba^{i_r}$ to $\Pi$ of length $2i_r+1$. (These last $2$ paths mirror those constructed in Case (\emph{b\/}).)\smallskip

Thus, our set $\Pi$ of potential shortest paths contains at most $r+1$ paths (from which at least one is a shortest path).

\subsection{Our algorithm}

We now use our set $\Pi$ of potential shortest paths so as to find a shortest path or the length of a shortest path. Our algorithm, DPillarMin, for finding the length of a shortest path in $G_{n,k}(0,x)$ is as follows. 
\begin{tabbing}
\hspace{0.2in}\={\tt \underline{Algorithm}:} DPillarMin\\
\>{\tt calculate $B$}\\
\>{\tt if $0\neq x$ then}\\
\>\hspace{0.2in}\={\tt $L = \min\{k+x,2k-x\}$}\\
\>\>{\tt calculate $r$, $s$, $\delta$, $\epsilon$, $\Delta_0$ and $\Delta_x$}\\
\>\>{\tt if $r=0$ then $L = \min\{L,x+\Delta_x\}$}\\
\>\>{\tt if $s=0$ then $L = \min\{L,k-x+\Delta_0\}$}\\
\>\>{\tt if $r=1$ then}\\
\>\>\hspace{0.2in}\={\tt $L = \min\{L,2k-2i_1+x-1+\Delta_x,$}\\
\>\>\>\hspace{1.7in}{\tt $2i_1-x+1\}$}\\
\>\>{\tt if $s=1$ then}\\
\>\>\>{\tt $L = \min\{L,k-2j_1+x-1+\Delta_0,$}\\
\>\>\>\hspace{1.7in}{\tt $k+2j_1-x+1\}$}\\
\>\>{\tt if $r\geq 2$ then}\\
\>\>\>{\tt calculate $\delta$}\hspace{0.1in}\% only need consider max. $\delta_l$\\
\>\>\>{\tt $L = \min\{L,2k-2\delta-x,$}\\
\>\>\>\hspace{0.65in}{\tt $2k-2i_1+x-1+\Delta_x,2i_r-x+1\}$}\\
\>\>{\tt if $s\geq 2$ then}\\
\>\>\>{\tt calculate $\epsilon$}\hspace{0.15in}\% only need consider max. $\epsilon_l$\\
\>\>\>{\tt $L = \min\{L, k-2\epsilon+x,k-2j_s+x-1+\Delta_0,$}\\
\>\>\>\hspace{1.8in}{\tt $k+2j_1-x+1\}$}\\
\>{\tt else}\\
\>\>{\tt calculate $r$ and $\delta$}\\
\>\>{\tt if $r=0$ then $L = 1$}\\
\>\>{\tt if $r=1$ then}\\
\>\>\>{\tt if $i_1=k-1$ then $L=1+\Delta_0$}\\
\>\>\>{\tt if $1=i_1\neq k-1$ then $L=3$}\\
\>\>\>{\tt if $1 \neq i_1 \neq k-1$ then}\\
\>\>\>\hspace{0.2in} $L=\min\{2k-2i_1-1+\Delta_0,2i_1+1\}$\\
\>\>{\tt if $r\geq 2$ then}\\
\>\>\>{\tt $L = \min\{k,2k-2\delta,2k-2i_1-1+\Delta_0,2i_r+1\}$}\\
\>{\tt output $L$}
\end{tabbing}

If we wish to output a shortest path then all we do is apply the algorithm DPillarMin but remember which shortest path corresponds to the final value of $L$ and output this shortest path (note that there may be more than one shortest path; exactly which path one obtains depends upon how one implements checking the paths of $\Pi$). The time complexity of both algorithms is clearly $O(k)$; that is, linear in the number of columns. Henceforth, we assume that the algorithm DPillarMin outputs an actual shortest path.

It should be clear (using Lemma~\ref{lem:turns}) that the different considerations for $r$ and $s$ exhaust all possibilities and that consequently the set of paths $\Pi$ considered by DPillarMin is such as to contain a shortest path. Hence, DPillarMin clearly outputs a shortest path from some source node to some destination node in DPillar$_{n,k}$. In summary, we have the following result.

\begin{theorem}
Suppose that $n,k\geq 2$ so that $n$ is even. The algorithm DPillarMin takes as input any two servers of DPillar$_{n,k}$, a source and a destination, and outputs a shortest path from the source server to the destination server; moreover, it computes this path with time complexity $O(k)$.
\end{theorem}

We can confirm that we have undertaken experiments so as to empirically check, using a breadth-first search, the correctness of DPillarMin on DPillar$_{n,k}$ when $n$ and $k$ are relatively small. We undertook our experiments using our in-house simulator \emph{INRFlow\/} \cite{INRFlow}.

\subsection{The diameter of DPillar}

We also compute the diameter of the DCN DPillar$_{n,k}$, \emph{i.e.}, the maximum of the lengths of shortest paths joining any two distinct servers. All that was stated in \cite{LYY12} was that the diameter of the DCN DPillar$_{n,k}$ is a `linear function of $k$'.

\begin{theorem}\label{thm:diam}
If $k\in\{2,3\}$ then the DCN DPillar$_{n,k}$ has diameter $k$; and if $k\geq 4$ then the DCN DPillar$_{n,k}$ has diameter $k+\lfloor\frac{k}{2}\rfloor-2$.
\end{theorem}

\begin{proof}Let $src$ and $dst$ be nodes of the digraph DPillar$_{n,k}$. W.l.o.g. we may assume that the column-index of $src$ is $0$ and that of $dst$ is $x$. We work in $G_{n,k}(0,x)$ and in the context of the algorithm DPillarMin.

We first note that for any $x$, the worst-case scenario is when all nodes of $G_{n,k}(0,x)$ are marked as a shortest path in this scenario yields a path in any other scenario (though not necessarily a shortest one). Hence, in what follows we assume that all nodes are marked.\medskip

\noindent\underline{Case (\emph{a\/})}: $k\geq 5$.\smallskip

\noindent We consider first the case when $x\neq 0$. There are $5$ different scenarios for $(r,s)$: $(0,\geq 2)$; $(1,\geq 2)$: $(\geq 2,\geq 2)$; $(\geq 2,1)$; and $(\geq 2,\geq 2)$.

Consider first when $r\geq 2$ and $s\geq 2$. By consideration of the algorithm DPillarMin, where we have $\delta=1$, $i_1=x+1$, $i_r=k-1$, $\epsilon=1$, $j_s=1$, $j_1=x-1$, $\Delta_0=1$ and $\Delta_x=1$, we immediately see that $L = \min\{k+x,2k-x,2k-2-x,2k-2(x+1)+x,2(k-1)-x+1,k-2+x,k-2+x,k+2(x-1)-x+1\} = \min\{k+x-2,2k-x-2\}$. We are trying to find a value of $x$ that maximizes this minimum value. If $k+x-2 \geq 2k-x-2$ then $x \geq \frac{k}{2}$; so, in this situation this minimum value is maximized when $x=\lceil\frac{k}{2}\rceil$ and this minimum value is then $2k-\lceil\frac{k}{2}\rceil-2 = k+\lfloor\frac{k}{2}\rfloor-2$. If $k+x-2 \leq 2k-x-2$ then $x \leq \frac{k}{2}$; so, in this situation this minimum value is maximized when $x=\lfloor\frac{k}{2}\rfloor$ and this minimum value is then $k+\lfloor\frac{k}{2}\rfloor-2$.

In each of the other $4$ cases for $(r,s)$, where $x\in\{0,1,k-2,k-1\}$, we have that the length of the path produced by DPillarMin is trivially less than $k+\lfloor\frac{k}{2}\rfloor-2$ (simply look at the initial minimization $L=\min\{k+x,2k-x\}$). Also, when $x=0$ the length of the path produced is trivially less than  $k+\lfloor\frac{k}{2}\rfloor-2$. Hence, when $k\geq 5$ the dameter is $k+\lfloor\frac{k}{2}\rfloor-2$.\smallskip

\noindent\underline{Case (\emph{b\/})} $2\leq k\leq 4$.\smallskip

\noindent It is trivial to see by hand that the diameter in this case is $k$. The result follows.
\end{proof}

\section{Experimental work}
Whilst we have obtained an optimal single-path routing algorithm for DPillar (optimal in that our algorithm always outputs a shortest path), as yet we have no idea as to how often the single-path routing algorithm DPillarSP is sub-optimal and the savings to be made by employing our optimal algorithm. To undertake a precise analytical evaluation of this question would be challenging; consequently, we proceed to evaluate empirically the most important performance metrics, namely path length, aggregate bottleneck throughput and transmission latency.

We undertake our evaluation using simulation. We use our own flow-based framework \emph{INRFlow\/} \cite{INRFlow}. The reason we adopt a simulation-based evaluation is as follows. Future DCNs are intended to incorporate hundreds of thousands, if not millions, of processors. Consequently, building a test-bed of servers (bearing in mind realistic access to resources) would only yield a DCN with a handful of servers and there would be no grounds for believing that any such evaluation would scale up. For instance, in order to build even the smallest meaningful DPillar$_{n,k}$ would require that $n$ should be at least $6$ and $k$ at least $3$ which would result in a test-bed with $81$ servers which is beyond our means. Not surprisingly, simulation is the standard evaluation mechanism in the literature. Of the DCNs mentioned in this paper, FiConn, MCube, HCN, BCN, SWKautz, SWCube, and SWdBruijn were all evaluated using simulation with only DCell, BCube, and CamCube evaluated using test-beds, incorporating $20$, $16$ and $27$ servers, respectively. In addition, the aspects of symmetry present in DCNs ameliorates the likelihood of `random' aspects of the network topology having an unexpected impact upon performance when compared with more unstructured networks. Finally, as regards our evaluation of communication latency in Section~\ref{sec:latency}, we have incorporated realistic measurements of protocol stack, propagation, data transmission, and routing latencies into our analysis.

\subsection{Path Length}

In order to obtain some idea of the practical significance of our algorithm DPillarMin in terms of path length, we undertook the following experiment. For specific values of $n$ and $k$, we measured the average length of the paths obtained by employing both DPillarMin and DPillarSP for every possible source-destination pair (node-symmetry means that we can actually fix a unique source node) as well as the cumulative frequencies of the lengths of paths arising. We also measured the number of such occasions when the path derived by DPillarSP is longer than the path derived by DPillarMin; that is, the number of times DPillarSP produced a non-minimal path. Our results are shown in Table~\ref{dbar} and Table~\ref{cfds}. In Table~\ref{dbar}, the columns denote (in order): the parameters $n$ and $k$ of the particular DPillar$_{n,k}$ that we are working with; the number of servers in that DPillar$_{n,k}$; the average path lengths obtained from inputting every possible source-destination pair to the algorithms DPillarMin and DPillarSP; the improvement in terms of average path length obtained by employing DPillarMin as a percentage of the average path length obtained by employing DPillarSP; and the percentage of source-destination pairs where the optimal path length is shorter than that obtained by employing DPillarSP. In Table~\ref{cfds}, for each chosen $n$ and $k$ we show the cumulative frequencies of the lengths of paths obtained by employing the two algorithms DPillarSP and DPillarMin. These cumulative frequencies are shown as percentages of the total number of pairs of (not necessarily distinct) servers and are rounded to the nearest $0.1$\% (in order to save space we do not show data relating to all pairs of $n$ and $k$; this omitted data is as might be expected).

\begin{table}[ht]
\caption{Average path lengths: DPillarMin vs. DPillarSP.}
\centering
\begin{tabular}{| c | c | c@{\hspace{3pt}} | @{\hspace{1pt}}c@{\hspace{1pt}} | @{\hspace{1pt}}c@{\hspace{1pt}} | @{\hspace{1pt}}c@{\hspace{1pt}} | @{\hspace{1pt}}c@{\hspace{1pt}} |}
\hline
\multicolumn{2}{| c |}{DPillar$_{n,k}$} & \# of & av. pth. len. & av. pth. len. & av. length & non-min.\\
\cline{1-2} $n$ & $k$ & servers & \hspace{3pt}DPillarMin\hspace{3pt} & \hspace{3pt}DPillarSP\hspace{3pt} & improve. & paths\\
\hline
16 & 3 & 1,536 & 2.72 & 3.86 & 29\% & 66\% \\
16 & 4 & 16,384 & 3.74 & 5.36 & 30\% & 73\% \\
16 & 5 & 163,840 & 4.77 & 6.86 & 30\% & 78\% \\\hline
32 & 3 & 12,288 & 2.86 & 3.93 & 27\% & 67\% \\
32 & 4 & 262,144 & 3.87 & 5.43 & 28\% & 74\% \\\hline
48 & 3 & 41,472 & 2.9 & 3.96 & 26\% & 67\% \\\hline
64 & 3 &  98,304  & 2.93 & 3.97 & 26\% & 67\% \\\hline
80 & 3 & 192,000 & 2.94 & 3.97 & 25\% & 67\% \\\hline
128 & 3 & 786,432 & 2.96 & 3.98 & 25\% & 67\%  \\\hline
\end{tabular}
\label{dbar}
\end{table}

\begin{table}[ht]
\caption{Cumulative frequencies of path lengths: DPillarMin vs. DPillarSP.}
\centering
\begin{tabular}{| c | c | c | @{\hspace{0pt}}c@{\hspace{2pt}} | @{\hspace{2pt}}c@{\hspace{2pt}} | @{\hspace{2pt}}c@{\hspace{2pt}} | @{\hspace{2pt}}c@{\hspace{2pt}} | @{\hspace{2pt}}c@{\hspace{2pt}} | @{\hspace{2pt}}c@{\hspace{2pt}} | @{\hspace{2pt}}c@{\hspace{2pt}} | @{\hspace{2pt}}c@{\hspace{2pt}} | @{\hspace{2pt}}c@{\hspace{2pt}} |@{\hspace{2pt}}c@{\hspace{2pt}} |}
\hline
\multicolumn{2}{| c |}{DPillar$_{n,k}$} & rout. & \multicolumn{10}{ c |}{path lengths}\\
\cline{1-2} \cline{4-13} $n$ & $k$ & alg. & \hspace{2pt}0 & 1 & 2 & 3 & 4 & 5 & 6 & 7 & 8 & 9\\
\hline
16 & 3 & SP & \hspace{2pt}0.1 & 0.6 & 4.8 & 38.0 & 70.8 & 100 & $-$ & $-$ & $-$ & $-$ 
\\
16 & 3 & Min & \hspace{2pt}0.1 & 2.0 & 26.2 & 100 & $-$ & $-$ & $-$ & $-$ & $-$ & $-$ 
\\\hline
16 & 5 & SP & \hspace{2pt}0.0 & 0.0 & 0.0 & 0.4 & 2.9 & 22.9 & 42.9 & 62.8 & 82.5 & 100 
\\
16 & 5 & Min & \hspace{2pt}0.0 & 0.0 & 0.3 & 2.5 & 20.3 & 100 & $-$ & $-$ & $-$ & $-$ 
\\\hline
32 & 4 & SP & \hspace{2pt}0.0 & 0.0 & 0.1 & 1.7 & 26.7 & 51.7 & 76.6 & 100 & $-$ & $-$ 
\\
32 & 4 & Min & \hspace{2pt}0.0 & 0.0 & 0.7 & 12.0 & 100 & $-$ & $-$ & $-$ & $-$ & $-$ 
\\\hline
80 & 3 & SP & \hspace{2pt}0.0 & 0.0 & 0.0 & 33.3 & 67.4 & 100 & $-$ & $-$ & $-$ & $-$ 
\\
80 & 3 & Min & \hspace{2pt}0.0 & 0.1 & 5.7 & 100 & $-$ & $-$ & $-$ & $-$ & $-$ & $-$ 
\\\hline
128 & 3 & SP & \hspace{2pt}0.0 & 0.0 & 0.5 & 33.9 & 67.2 & 100 & $-$ & $-$ & $-$ & $-$ 
\\
128 & 3 & Min & \hspace{2pt}0.0 & 0.0 & 3.6 & 100 & $-$ & $-$ & $-$ & $-$ & $-$ & $-$ 
\\\hline

\end{tabular}
\label{cfds}
\end{table}

As can be seen from Table~\ref{dbar}, using the algorithm DPillarMin yields a very significant improvement of between $25\%$ and $30\%$ in terms of the average path length. It is also worth highlighting that the number of non-optimal paths generated by DPillarSP is between $66\%$ and $78\%$ and increases significantly with $k$. Note that a reduction in path length does not only mean that the latency experienced by network traffic should be reduced (more on this later) but also that each flow will require less aggregate bandwidth to be transmitted and so the overall throughput of the network should also increase. As can be seen from Table~\ref{cfds}, in each of the chosen scenarios DPillarMin yields significant cumulative improvements in path length. For example, when $n=16$ and $k=5$, with DPillarSP only 22.9\% of all paths have length at most $5$ whereas with DPillarMin all paths do. We measure next the aggregate bottleneck throughput obtained through using the two different routing algorithms.

\subsection{Aggregate Bottleneck Throughput}
The \emph{aggregate bottleneck throughput}, or simply \emph{ABT}, is a metric introduced in \cite{GLL09} in order to estimate the throughput performance of a DCN. The reasoning behind ABT is that the performance of an all-to-all operation is limited by its slowest flow, \emph{i.e.}, the flow with the lowest throughput.  The ABT is defined as the total number of flows times the throughput of the \emph{bottleneck flow}, \emph{i.e.}, the link sustaining the most flows. In our experiments the bottleneck flow is determined experimentally using actual routing functions within our framework \emph{INRFlow\/}. We assume an all-to-all communication pattern, so that there are $N(N - 1)$ flows, and a bandwidth of $1$ unit per directional link, where $N$ is the total number of servers. Since datacenters are most commonly used as stream processing platforms and are therefore bandwidth limited, this is an extremely relevant performance metric.

\begin{table}[ht]
\caption{Aggregate bottleneck throughput: DPillarMin vs. DPillarSP.}
\centering
\begin{tabular}{| c | c| c@{\hspace{3pt}} | @{}c@{} | @{}c@{} | c |}
\hline
\multicolumn{2}{| c |}{DPillar$_{n,k}$} & \# of & ABT & ABT & ABT \\
\cline{1-2} $n$ & $k$ & servers & \hspace{3pt}DPillarMin\hspace{3pt} & \hspace{3pt}DPillarSP\hspace{3pt} & improve. \\
\hline
16	&	3	&	1,536	&	757.16	&	397.93	&	90\%	\\
16	&	4	&	16,384	&	6077.88	&	3056.72	&	99\%	\\
16	&	5	&	163,840	&	52953.26	&	23883.38	&	122\%	\\\hline
32	&	3	&	12,288	&	5651.85	&	3126.72	&	81\%	\\
32	&	4	&	262,144	&	92102.69	&	48276.98	&	91\%	\\\hline
48	&	3	&	41,472	&	18634.09	&	10472.73	&	78\%	\\\hline
64	&	3	&	98,304	&	43653.56	&	24761.71	&	76\%	\\\hline
80	&	3	&	192,000	&	84659.97	&	48362.72	&	75\%	\\\hline
128	&	3	&	786,432	&	343097.99	&	197595.98	&	74\%	\\\hline
\end{tabular}
\label{abt}
\end{table}

Table \ref{abt} shows that DPillarMin is capable of offering much higher ABT than DPillarSP in all cases, with improvements of between 74\% and 122\% (the right-most is the improvement in ABT by using DPillarMin rather than DPillarSP divided by the ABT of DPillarSP). Informally, this means that bandwidth-limited  applications such as, for example, Big Data analytics, running over a DCN using DPillarMin might be able to achieve nearly twice as much computational throughput as the same application running over a DCN using DPillarSP. This can provide significant savings in terms of running and maintenance costs associated with each application and thus will result in more competitive pricing for tenants. Furthermore, as applications run faster it will be possible to run more applications in a given time frame and so there is a huge potential for increasing the overall profit of the datacenter.

\subsection{Communication Latency} \label{sec:latency}
Not all datacenter applications are bandwidth sensitive; indeed, many of them are more sensitive to latency, such as real-time operations or, more generally, any application interfacing with users. For this reason, it is important that we look at the transmission latency that we can expect from DPillarSP and DPillarMin. As there is no server-centric DCN framework available that will enable us to perform testbench experiments (building one ourselves is not possible), we measure the latencies imposed by the different steps of the transmission, namely within the protocol stack, propagation latency, data transmission latency and routing at the servers, so as to obtain an estimate of how changing the routing algorithm would affect the overall performance. Our experiments were as follows.

\begin{itemize}
	\item We measured the round trip time of both an \emph{empty} frame (28 bytes for the headers) and a \emph{full} frame (1,500 bytes, including the headers) sent to localhost so as to measure the latency imposed by going up and down the protocol stack. In both cases, the \emph{stack latency}, $L_{s}$, was found to be 10 $\mu$s.
	\item We measured the round trip of an empty frame sent to a neighbouring server connected to the same Gigabit Ethernet switch. This was found to be 64 $\mu$s; thus we can compute the one-way \emph{propagation latency}, $L_p$, \emph{i.e.}, the time to go through the links and the switch, by dividing by two and removing the stack latency. This yields a propagation latency of 22 $\mu$s.
	\item We measured the round trip time of a full-frame sent to the same neighbouring server. This was found to be 140 $\mu$s; thus the one-way \emph{data transfer latency}, $L_{d}$, can be calculated similarly by dividing by two and subtracting the stack latency as well as the propagation latency. This results in a data transfer latency of 38 $\mu$s (roughly 26 ns per byte).
	\item We measured the \emph{average routing latency}, $L_{r}$, for both algorithms for a selection of the configurations above (those with between 8 and 40 thousand servers). Note that the code for the two algorithms is not optimised and that it includes some overheads imposed by our framework; so the running times for the algorithms can be considered as a worst case. 
\end{itemize}
Consequently, for both DPillarSP and DPillarMin we obtain the \emph{per-hop latency} $L_{hop} = L_s + L_p + L_d + L_r$ along with the \emph{server-to-server} latency $L_{total} = L_{hop}\times \bar{d}$, where $\bar{d}$ is the average path length.

All the measurements were carried out under low load conditions in the same server: a 32-core AMD Opteron 6220 with 256 Gbytes of RAM and running an Ubuntu 14.04.1 SMP OS. Round-trip time measurements were carried out with the ping utility. The server and its neighbour are located within the same rack and are connected with short (at most $1$ metre) electrical wires to a 24-port 1-Gbit Ethernet switch which does not support jumbo frames (we do not have 10-Gbit Ethernet hardware available). Note that the use of short wires is the best case for the propagation delay, as in a real scale-out datacenter wires will be much longer and so propagation delays will be larger (even if fibre connections are used~\cite{PD11}). Similarly, the measured latency of the protocol stack does not take into account any extra management/control inherent to the server-centric nature of the system; so, again, it can be considered a best case scenario. Increasing these delays will dilute the effects of the average routing latency in the total latency even more than in our preliminary estimate. (We remind the reader that a DPillar datacenter would be constructed out of COTS hardware and so our experimental set-up is reasonable).

\begin{table}[ht]
\caption{Average routing latencies: DPillarMin vs. DPillarSP.}
\centering
\begin{tabular}{| c | c | c@{\hspace{1pt}} | @{\hspace{1pt}}c@{\hspace{1pt}} | @{\hspace{1pt}}c@{\hspace{1pt}} | @{\hspace{2pt}}c@{\hspace{2pt}} |}
\hline
\multicolumn{2}{| c |}{DPillar$_{n,k}$} & \# of & $L_{r}$ & $L_{r}$ & $L_{r}$ \\
\cline{1-2} $n$ & $k$ & servers & \hspace{3pt}DPillarMin\hspace{3pt} & \hspace{3pt}DPillarSP\hspace{3pt} & increase  \\
\hline
16	&	4	&	16,384	&	5.964 $\mu$s	&	1.349 $\mu$s	&	442\% \\
32	&	3	&	12,288	&	3.325 $\mu$s	&	0.960 $\mu$s	&	346\% \\
48	&	3	&	41,472	&	3.328 $\mu$s	&	0.859 $\mu$s	&	387\% \\
\hline
\end{tabular}
\label{runtime}
\end{table}

\begin{table}[ht]
\caption{Per-hop and overall latencies: DPillarMin vs. DPillarSP.}
\centering
\begin{tabular}{| @{\hspace{3pt}}c@{\hspace{3pt}} | @{\hspace{1pt}}c@{\hspace{1pt}} | c@{} | @{}c@{} | @{\hspace{1pt}}c@{\hspace{1pt}} | @{}c@{} | @{}c@{} | @{\hspace{1pt}}c@{\hspace{1pt}} |}
\hline
\multicolumn{2}{| @{}c@{} |}{DPillar$_{n,k}$} & $L_{hop}$ & $L_{hop}$ & $L_{hop}$ & $L_{total}$ & $L_{total}$ & $ L_{total}$ \\
\cline{1-2} $n$ & $k$ & \hspace{2pt}\scriptsize DPillarMin\hspace{2pt} & \hspace{2pt}\scriptsize DPillarSP\hspace{2pt} & decl. & \hspace{2pt}\scriptsize DPillarMin\hspace{2pt} & \hspace{2pt} \scriptsize DPillarSP\hspace{2pt} & improve. \\
\hline
16	&	4	&	76.0	&	71.3	&	6\%	&	284.1	&	382.2	&	26\%	\\
32	&	3	&	73.3	&	71.0	&	3\%	&	209.5	&	279.1	&	25\%	\\
48	&	3	&	73.3	&	70.9	&	3\%	&	212.9	&	280.4	&	24\%	\\
\hline
\end{tabular}
\label{latency}
\end{table}

Table~\ref{runtime} shows the average routing latency $L_r$ for DPillarSP and DPillarMin, along with the increase in the average routing latency when using DPillarMin as opposed to DPillarSP (shown as a percentage of the average routing latency when using DPillarSP). Table~\ref{latency} details the per-hop and server-to-server latencies for both DPillarSP and DPillarMin. The very slight increase in the per-hop latency when using DPillarMin as opposed to DPillarSP is shown, as is the improvement in the server-to-server latency when using DPillarMin as opposed to DPillarSP (both are shown as a percentage of the corresponding value for DPillarSP). 

It can be seen that the average routing latency for DPillarMin is between 3.4 and 4.5 times slower than that for DPillarSP, but of the order of only a few microseconds which is well below the other latencies measured in our experimental set-up. In consequence, the per-hop latencies of DPillarSP and DPillarMin are very similar; however, there is a significant reduction in server-to-server latency for DPillarMin over DPillarSP (between 24\% and 26\%) when the reductions in average path-length are factored in. 

Informal analytical modelling using the values above as a reference suggests that if jumbo frames were used then there would be negligible increase in per-hop latency so as to yield a significant overall improvement in server-to-server latency of up to 30\%. A similar analytic analysis using 10-Gbit Ethernet hardware suggests that while per-hop latency can increase by up to 12.5\% when using DPillarMin as opposed to DPillarSP, the overall server-to-server latency improvement will still be in the range 19-23\%. Finally, the estimates with jumbo frame-enabled 10-Gbit Ethernet yield very similar results as the ones presented here. Full details are available in the supplemental material.

While the latency analysis performed here is rather simplistic and only covers zero-load latencies, our objective is not to provide highly accurate latency figures but to show that the impact of the routing algorithm DPillarMin on latency is insignificant, particularly when compared with the huge gains in terms of path length and throughput. Note that due to its less favourable throughput, the use of DPillarSP would lead to additional queuing in the servers which would in have a detrimental impact upon performance.

\subsection{FiConn and DCell}

There does not exist a proper comparative experimental evaluation of the numerous (dual-port) server-centric DCNs in the literature; comparative evaluations that have been undertaken so far are somewhat \emph{ad hoc}, both in terms of the DCNs compared and the performance metrics evaluated. Of course, an extensive comparative evaluation will be a significant body of work and is well beyond the scope of this paper (where our focus has been on improving routing in DPillar). Moreover, we are fully aware that there are many different metrics for DCN evaluation, such as those relating to fault-tolerance, bisection bandwidth, load balancing, latency, throughput, scalability, and so on, and that the eventual success of a DCN will usually depend on its capacity to cope well across a range of such metrics. Nevertheless, we end our experimentation by including an interesting prelude to a fuller analysis of routing within server-centric DCNs: we briefly compare routing in DPillar with routing in the two DCNs DCell and FiConn.

We have chosen FiConn and DCell as they are widely regarded as benchmark server-centric DCNs. Like DPillar, FiConn is  dual-port, whereas DCell is such that the number of server NIC ports is variable. The reader is referred to \cite{LGW09} and \cite{GWT08} for definitions of FiConn and DCell, respectively, but just as with DPillar, FiConn and DCell are families of DCNs parameterized by $n$, the number of switch-ports in a switch, and $k$, the depth of the recursive construction (actually, a server in DCell$_{n,k}$ has $k+1$ NIC ports).

In Table~\ref{dpfcdc}, we have displayed the average path length and the ABT of (various instantiations of) DPillar with the routing algorithm DPillarMin, FiConn with the routing algorithm TOR (from \cite{LGW09}), and DCell with the routing algorithm DCellRouting (from \cite{GWT08}); we have chosen these instantiations so that the different DCNs can be compared on three different bases, namely them all having roughly 24K, 117K and 170K servers, respectively (so, we have normalized against the number of servers). As usual, the data in Table~\ref{dpfcdc} has been derived using our tool INRFlow. 

\begin{table}[ht]
\caption{Average path lengths and ABT: DPillarMin vs. FiConn vs. DCell}
\centering
\begin{tabular}{| c | c | c@{\hspace{3pt}} | @{\hspace{1pt}}c@{\hspace{1pt}}  | @{\hspace{1pt}}c@{\hspace{1pt}}  |}
\hline
\multicolumn{2}{| c |}{DPillar$_{n,k}$} & \# of & av. pth. len. & ABT \\
\cline{1-2} $n$ & $k$ & servers & \hspace{3pt}DPillarMin\hspace{3pt} & \hspace{3pt}DPillarMin\hspace{3pt}\\
\hline
12 & 5 & 38,880 & 4.68 & 12805.63 \\
16 & 5 & 163,840 & 4.77 & 52952.94\\
18 & 4 & 26,244 & 3.77 & 9616.46\\
26 & 4 & 114,244 & 3.84 & 40637.47 \\
48 & 3 & 41,472 & 2.90 & 18633.64\\
64 & 3 &  98,304 & 2.93 & 43653.12\\
\hline
\multicolumn{2}{| c |}{FiConn$_{n,k}$} & \# of & av. pth. len. & ABT \\
\cline{1-2} $n$ & $k$ & servers & \hspace{3pt}TOR\hspace{3pt} & \hspace{3pt}TOR\hspace{3pt}\\
\hline
10 & 3 & 116,160 & 12.97 & 13026.18 \\
24 & 2 & 24,648 & 6.56 & 5005.47 \\
36 & 2 &  117,648  &  6.71   & 23694.75 \\
40 & 2 & 177,240 & 6.74 & 35650.59\\
\hline
\multicolumn{2}{| c |}{DCell$_{n,k}$} & \# of & av. pth. len. & ABT \\
\cline{1-2} $n$ & $k$ & servers & \hspace{3pt}DCellRouting\hspace{3pt} & \hspace{3pt}DCellRouting\hspace{3pt}\\
\hline
3 & 3 & 24,492 & 10.18 & 5475.43 \\
4 & 3 & 176,820 & 11.29 & 33582.97\\
12 & 2 & 24,492 & 6.34 &  6968.73\\
18 & 2 &  117,306  &  6.56   & 31937.10\\
\hline
\end{tabular}
\label{dpfcdc}
\end{table}

As can readily be seen, DPillar compares extremely well with FiConn and DCell in terms of both average path length and ABT (even though DCell would appear to have a natural advantage over the other two DCNs as it involves servers with more than two NIC ports).

We end with some comments on our very brief comparative evaluation. First, we reiterate that what is really required is an extensive evaluation involving a range of server-centric DCNs across a range of performance metrics. Second, we observe (in comparison with data in Tables~\ref{dbar}~and~\ref{abt}) that the improvements made in using DPillarMin in the DCN DPillar, rather than DPillarSP, have resulted in moving DPillar from only comparable with DCell and FiConn to better than DCell and FiConn (at least in terms of average paths length and ABT). Third, there is no reason why a closer combinatorial scrutiny of both DCell and FiConn might not result in new and better routing algorithms than DCellRouting and TOR, respectively (just as we have improved routing within DPillar within this paper).

\section{Conclusions}\label{sec:conclusions}

In this paper we have: developed an optimal and practical single-path routing algorithm DPillarMin for the DCN DPillar; shown that DPillar is a Cayley graph, and so node-symmetric; and provided an exact formulation of the diameter of DPillar. Our experimental results show not only that DPillarMin can significantly reduce the average path length of network traffic (up to 30\%), but also that this reduction results in a significant increase (more than 2$\times$) in terms of overall network throughput. Finally we showed that the computational overhead of DPillarMin is negligible and will barely affect the processing of network traffic: less than a 6\% increase in per-hop latency, which is more than compensated by the reductions in path length.

In summary, we can claim that our proposed routing algorithm can unleash a massively improved performance to the DPillar DCN. Furthermore, we feel that there are other areas where efficiency gains might be made; in particular, with regard to multi-path routing. Of course, we reaffirm our statement above that what also needs to be undertaken is an holistic comparison of different (dual-port) server-centric DCNs, with their different routing algorithms and across a wide range of performance metrics, along with the combinatorial study of DCNs different to DPillar with a view to improving their routing algorithms.

\section*{Acknowledgements}

All of the authors are supported by the EPSRC grants EP/K015680/1 and EP/K015699/1 `Interconnection Networks: Practice unites with Theory (INPUT)'. Dr. Javier Navaridas is supported by the European Union's Horizon 2020 programme under grant agreement No. 671553 `ExaNeSt'. We are also grateful for the insightful comments of the reviewers which helped to significantly improve this paper.





\appendix
\renewcommand{\thesection}{A}
\setcounter{figure}{0} \renewcommand{\thefigure}{A.\arabic{figure}}
\setcounter{table}{0} \renewcommand{\thetable}{A.\arabic{table}}

\section*{Appendix}
\subsection{Proof of Lemma 1}

\begin{lemma}
The digraph DPillar$_{n,k}$ is a Cayley graph.
\end{lemma}

\begin{proof}
We first define a set of elements and then a notion of multiplication on this set. Let $t_0,t_1,\ldots,t_{k-1}$ be distinct symbols and for each $i\in\{0,1,\ldots,k-1\}$, define
\begin{align*}
G_i = \{t_i^{p_i}t_{i+1}^{p_{i+1}}\ldots t_{k-1}^{p_{k-1}}t_0^{p_0}t_1^{p_1}\ldots t_{i-1}^{p_{i-1}} : 0\leq p_j\leq \frac{n}{2}-1, \text{ for all }j\in\{0,1,\ldots,k-1\}\}.
\end{align*}
Define $G_k^n = \cup_{i=0}^{k-1}G_i$; so, note that $|G_k^n| = k(\frac{n}{2})^k$. Define
\begin{align*}
S_a &= \{t_{k-1}^{q}t_{0}^{0} t_{1}^{0}\ldots t_{k-2}^{0} : 0\leq q\leq \frac{n}{2}-1\};\\
S_b &= \{t_{0}^{q}t_{1}^{0} \ldots t_{k-1}^{0} : 0\leq q\leq \frac{n}{2}-1\};\\
S_c &= \{t_1^{0}t_{1}^{0}\ldots t_{k-1}^{0}t_0^{q} : 0\leq q\leq \frac{n}{2}-1\};\\
S_d &= \{t_{0}^{0}t_{1}^{0} \ldots t_{k-1}^{q} : 0\leq q\leq \frac{n}{2}-1\}.
\end{align*}
Define a right-multiplication $\circ$ of the elements of $G_k^n$ by the elements of $S = S_a\cup S_b\cup S_c\cup S_d$ as follows:
\begin{align*}
t_i^{p_i}t_{i+1}^{p_{i+1}}\ldots t_{k-1}^{p_{k-1}}t_0^{p_0}t_1^{p_1}\ldots t_{i-1}^{p_{i-1}} \circ t_{k-1}^{q}t_{0}^{0} t_{1}^{0}\ldots t_{k-2}^{0} &= t_{i-1}^{p_{i-1}+ q}t_i^{p_i}t_{i+1}^{p_{i+1}}\ldots t_{k-1}^{p_{k-1}}t_0^{p_0}t_1^{p_1}\ldots t_{i-2}^{p_{i-2}}\\
t_i^{p_i}t_{i+1}^{p_{i+1}}\ldots t_{k-1}^{p_{k-1}}t_0^{p_0}t_1^{p_1}\ldots t_{i-1}^{p_{i-1}} \circ t_{0}^{q}t_{1}^{0} \ldots t_{k-1}^{0} &= t_i^{p_i+q}t_{i+1}^{p_{i+1}}\ldots t_{k-1}^{p_{k-1}}t_0^{p_0}t_1^{p_1}\ldots t_{i-1}^{p_{i-1}}\\
t_i^{p_i}t_{i+1}^{p_{i+1}}\ldots t_{k-1}^{p_{k-1}}t_0^{p_0}t_1^{p_1}\ldots t_{i-1}^{p_{i-1}} \circ t_1^{0}t_{2}^{0}\ldots t_{k-1}^{0}t_0^{q} &= t_{i+1}^{p_{i+1}}\ldots t_{k-1}^{p_{k-1}}t_0^{p_0}t_1^{p_1}\ldots t_{i-1}^{p_{i-1}}t_i^{p_i+q}\\
t_i^{p_i}t_{i+1}^{p_{i+1}}\ldots t_{k-1}^{p_{k-1}}t_0^{p_0}t_1^{p_1}\ldots t_{i-1}^{p_{i-1}} \circ t_{0}^{0}t_{1}^{0} \ldots t_{k-1}^{q} &= t_i^{p_i}t_{i+1}^{p_{i+1}}\ldots t_{k-1}^{p_{k-1}}t_0^{p_0}t_1^{p_1}\ldots t_{i-1}^{p_{i-1}+q},
\end{align*}
with $i$, $p_0$, $p_1$, $\ldots$, $p_{k-1}$ and $q$ as appropriate and with addition on superscripts modulo $\frac{n}{2}$. 

We now extend the multiplication we have just defined so that we make $G_n^k$ into a group with a generating set $S_0$ that is a subset of $S$. Let $\langle S\rangle$ be the set of all elements generated by right-multiplication by elements of $S$. It is trivial to show that this set is $G_k^n$; that is, 
\begin{align*}
G_k^n = \{((\ldots((s_1\circ s_2)\circ s_3) \ldots ) \circ s_i) : i \geq 1, s_j\in S \text{ for } j=1,2,\ldots,i\}.
\end{align*}
Extend the multiplication $\circ$ to $G_k^n$ by defining that no matter how a
multiplication of elements of $S$ is bracketed, \emph{e.g.}, as $(s_1 \circ
(s_2 \circ s_3)) \circ (s_4 \circ s_5)$, the product is defined as that obtained
by multiplying on the right, \emph{e.g.}, as $((((s_1\circ s_2) \circ s_3)
\circ s_4) \circ s_5)$. Consequently, we have now equipped $G_k^n$ with an
associative multiplication $\circ$. It is trivial to check that there is an
identity in $G_k^n$ (w.r.t. $\circ$; it is $t_0^0t_1^0\ldots t_{k-1}^0$) and
also that every element of $G_k^n$ has an inverse; furthermore, every element of
$S$ has an inverse in $S$. Hence, $G_k^n$ is a group generated by the $2n-2$
elements of $S_0 = S\setminus\{t_0^0t_1^0\ldots t_{k-1}^0\}$ and $S_0$ is closed
under inverses. Let $G_k^n(S)$ be the Cayley graph of $G_k^n$ w.r.t.\ the
generating set $S_0$.

Finally, we prove that the Cayley graph $G_k^n(S_0)$ is exactly the same as DPillar$_{n,k}$. In what follows, by DPillar$_{n,k}$ we mean the digraph DPillar$_{n,k}$. Define the mapping $\varphi$ from the nodes of $G_k^n(S_0)$ to the nodes of DPillar$_{n,k}$ by
$$\varphi(t_i^{p_i}t_{i+1}^{p_{i+1}}\ldots t_{k-1}^{p_{k-1}}t_0^{p_0}t_1^{p_1}\ldots t_{i-1}^{p_{i-1}}) = (i,p_{k-1}p_{k-2}\ldots p_0)$$
where $i$, $p_0$, $p_1$, $\ldots$, $p_{k-1}$ are as appropriate. As
\begin{align*}
  t_i^{p_i}t_{i+1}^{p_{i+1}}\ldots t_{k-1}^{p_{k-1}}t_0^{p_0}t_1^{p_1}\ldots t_{i-1}^{p_{i-1}} \circ t_{k-1}^{q}t_{0}^{0} t_{1}^{0}\ldots t_{k-2}^{0}
   = t_{i-1}^{p_{i-1}+q}t_i^{p_i}t_{i+1}^{p_{i+1}}\ldots t_{k-1}^{p_{k-1}}t_0^{p_0}t_1^{p_1}\ldots t_{i-2}^{p_{i-2}},
\end{align*}
this describes the a-edge of DPillar$_{n,k}$ from $(i,p_{k-1}p_{k-2}\ldots p_{0})$ to $(i-1,p_{k-1}p_{k-2}\ldots p_{i}(p_{i-1}+q)p_{i-2}\dots p_{0})$. As
\begin{align*}
t_i^{p_i}t_{i+1}^{p_{i+1}}\ldots t_{k-1}^{p_{k-1}}t_0^{p_0}t_1^{p_1}\ldots t_{i-1}^{p_{i-1}} \circ t_{0}^{q}t_{1}^{0} \ldots t_{k-1}^{0}
  = t_i^{p_i+q}t_{i+1}^{p_{i+1}}\ldots t_{k-1}^{p_{k-1}}t_0^{p_0}t_1^{p_1}\ldots t_{i-1}^{p_{i-1}},
\end{align*}
this describes the b-edge of DPillar$_{n,k}$ from $(i,p_{k-1}p_{k-2}\ldots
p_{0})$ to $(i,p_{k-1}p_{k-2}\ldots p_{i+1}(p_i+q)p_{i-1}\ldots p_0)$ when $q
>0$. As
\begin{align*}
  t_i^{p_i}t_{i+1}^{p_{i+1}}\ldots t_{k-1}^{p_{k-1}}t_0^{p_0}t_1^{p_1}\ldots t_{i-1}^{p_{i-1}} \circ t_1^{0}t_{1}^{0}\ldots t_{k-1}^{0}t_0^{q}
  = t_{i+1}^{p_{i+1}}\ldots t_{k-1}^{p_{k-1}}t_0^{p_0}t_1^{p_1}\ldots
  t_{i-1}^{p_{i-1}}t_i^{p_i+q},
\end{align*}
this describes the c-edge of DPillar$_{n,k}$ from $(i,p_{k-1}p_{k-2}\ldots p_{0})$ to $(i+1,p_{k-1}p_{k-2}\ldots p_{i+1}(p_i+q)p_{i-1}\dots p_{0})$. As
\begin{align*}
t_i^{p_i}t_{i+1}^{p_{i+1}}\ldots t_{k-1}^{p_{k-1}}t_0^{p_0}t_1^{p_1}\ldots t_{i-1}^{p_{i-1}} \circ t_{0}^{0}t_{1}^{0} \ldots t_{k-1}^{q}
  = t_i^{p_i}t_{i+1}^{p_{i+1}}\ldots t_{k-1}^{p_{k-1}}t_0^{p_0}t_1^{p_1}\ldots t_{i-1}^{p_{i-1}+q},
\end{align*}
this describes the d-edge of DPillar$_{n,k}$ from $(i,p_{k-1}p_{k-2}\ldots p_{0})$ to $(i,p_{k-1}p_{k-2}\ldots p_i(p_{i-1}+q)p_{i-2}\dots p_{0})$ when $q>0$. Consequently, $\varphi$ is an isomorphism of $G_k^n(S_0)$ to DPillar$_{n,k}$ and the result follows.\end{proof}

\subsection{Modelling other network configurations: jumbo frames and 10-Gbit Ethernet}
While we were unable to carry out experiments with more advanced datacenter
networking equipment, such as switches capable of dealing with jumbo frames or
10-Gbit Ethernet NICs or switches, it should be possible to extrapolate their
performance from the statistics we captured from our experimental set-up. Having
estimates for the latency expected from these configurations is useful as they
can be seen as pathological cases in relation to the performance gains inherent
to DPillarMin. Using jumbo frames (that is, frames with a payload of 9,000
bytes\footnote{Different networking equipment may have different frame length
  limits. For simplicity, we stick to a payload of 9,000 bytes, even though some
  devices can handle even larger jumbo frames, \emph{e.g.}, Cisco devices can
  typically handle up to 9,216 bytes.}, rather than the standard 1,472 bytes)
means that any routing algorithm is executed less often and that the protocol-
and propagation-induced delays become less substantial when compared with the
data transmission delay. In our case, this means that the overhead due to using
DPillarMin becomes less significant; consequently, the overall delay will be
even better than with standard frames. On the other hand, the higher bandwidth
of 10-Gbit (10$\times$) equipment means that the per-hop delay will be reduced
which, in turn, means that the time taken to undertake routing computations may
become dominant. However, according to our assessment this will not be the case.

We now explain how we extrapolate the per-hop latency and server-to-server latency for these technologies from the latencies we measured empirically in Section 7.3 (that is, $L_s$, $L_p$, $L_d$ and $L_r$).
\begin{itemize}
	\item The stack latency, $L_s$, should not change as it is due to software executions at the server-side. 
	\item The propagation latency, $L_p$, would barely be affected by the bandwidth of the links, or the size of the frames, but would be affected by the length of the links or the transmission media used (copper/fibre). For simplicity, we assume the propagation latency does not vary.
	\item The data transfer latency depends on the transmission bandwidth and the size of the frame. For simplicity, we assume perfect linear scaling of the per-byte delay: 26 ns per byte for 1-Gbit Ethernet and 2.6 ns per byte for 10-Gbit Ethernet, multiplied by either 1 standard frame (1,472 bytes) or 1 jumbo frame (9,000 bytes).
	\item There is no change to the average routing latency, $L_r$.
\end{itemize}

The extrapolation for 10-Gbit Ethernet (see Table \ref{10gbeth}) suggests that, even though the per-hop delay might go up significantly with faster networks, the great improvement in path length achieved by DPillarMin still compensates for this and provides an improvement in terms of overall latency of between 20\% and 23\%. The use of jumbo frames alleviates the overhead incurred by using DPillarMin (see Tables \ref{jumbo1gbeth} and \ref{jumbo10gbeth}) and raises the improvement in terms of overall latency up to 29\% (1-Gbit Ethernet) and 24\% (10-Gbit Ethernet).

Further informal analysis using stack and propagation delays that were one order of magnitude smaller than the ones obtained in our empirical testing, suggested that with standard frames the overall latency will still be reduced by around 20 to 23\% with 1-Gbit Ethernet and between 4\% and 16\% with 10-Gbit Ethernet in most of the cases. If jumbo frames were considered then the stack and propagation delays barely affect the overall latency so the figures remain similar to those discussed above.

\begin{table}[ht]
\caption{Per-hop and overall latencies with 10-Gbit Ethernet and standard frames.}
\centering
\begin{tabular}{| @{\hspace{3pt}}c@{\hspace{3pt}} | @{\hspace{1pt}}c@{\hspace{1pt}} | c@{} | @{}c@{} | @{\hspace{1pt}}c@{\hspace{1pt}} | @{}c@{} | @{}c@{} | @{\hspace{1pt}}c@{\hspace{1pt}} |}
\hline
\multicolumn{2}{| @{}c@{} |}{DPillar$_{n,k}$} & $L_{hop}$ & $L_{hop}$ & $L_{hop}$ & $L_{total}$ & $L_{total}$ & $L_{total}$ \\
\cline{1-2} $n$ & $k$ & \hspace{2pt}\scriptsize DPillarMin\hspace{2pt} & \hspace{2pt}\scriptsize DPillarSP\hspace{2pt} & decl. & \hspace{2pt}\scriptsize DPillarMin\hspace{2pt} & \hspace{2pt} \scriptsize DPillarSP\hspace{2pt} & improve. \\
\hline
16	&	4	&	41.76	&	37.15	&	12\%	&	156.2	&	199.0	&	22\%	\\
32	&	3	&	39.12	&	36.76	&	6\%	&	111.8	&	144.6	&	23\%	\\
48	&	3	&	39.13	&	36.66	&	7\%	&	113.6	&	145.0	&	22\%	\\
\hline
\end{tabular}
\label{10gbeth}
\end{table}

\begin{table}[ht]
\caption{Per-hop and overall latencies with 1-Gbit Ethernet and jumbo frames.}
\centering
\begin{tabular}{| @{\hspace{3pt}}c@{\hspace{3pt}} | @{\hspace{1pt}}c@{\hspace{1pt}} | c@{} | @{}c@{} | @{\hspace{1pt}}c@{\hspace{1pt}} | @{}c@{} | @{}c@{} | @{\hspace{1pt}}c@{\hspace{1pt}} |}
\hline
\multicolumn{2}{| @{}c@{} |}{DPillar$_{n,k}$} & $L_{hop}$ & $L_{hop}$ & $L_{hop}$ & $L_{total}$ & $L_{total}$ & $L_{total}$ \\
\cline{1-2} $n$ & $k$ & \hspace{2pt}\scriptsize DPillarMin\hspace{2pt} & \hspace{2pt}\scriptsize DPillarSP\hspace{2pt} & decl. & \hspace{2pt}\scriptsize DPillarMin\hspace{2pt} & \hspace{2pt} \scriptsize DPillarSP\hspace{2pt} & improve. \\
\hline
16	&	4	&	270.30	&	265.69	&	2\%	&	1011.0	&	1423.3	&	29\%	\\
32	&	3	&	267.66	&	265.30	&	1\%	&	764.6	&	1043.5	&	27\%	\\
48	&	3	&	267.67	&	265.20	&	1\%	&	777.3	&	1049.3	&	26\%	\\
\hline
\end{tabular}
\label{jumbo1gbeth}
\end{table}

\begin{table}[ht]
\caption{Per-hop and overall latencies with 10-Gbit Ethernet and jumbo frames.}
\centering
\begin{tabular}{| @{\hspace{3pt}}c@{\hspace{3pt}} | @{\hspace{1pt}}c@{\hspace{1pt}} | c@{} | @{}c@{} | @{\hspace{1pt}}c@{\hspace{1pt}} | @{}c@{} | @{}c@{} | @{\hspace{1pt}}c@{\hspace{1pt}} |}
\hline
\multicolumn{2}{| @{}c@{} |}{DPillar$_{n,k}$} & $L_{hop}$ & $L_{hop}$ & $L_{hop}$ & $L_{total}$ & $L_{total}$ & $L_{total}$ \\
\cline{1-2} $n$ & $k$ & \hspace{2pt}\scriptsize DPillarMin\hspace{2pt} & \hspace{2pt}\scriptsize DPillarSP\hspace{2pt} & decl. & \hspace{2pt}\scriptsize DPillarMin\hspace{2pt} & \hspace{2pt} \scriptsize DPillarSP\hspace{2pt} & improve. \\
\hline
16	&	4	&	61.20	&	56.58	&	8\%	&	228.9	&	303.1	&	24\%	\\
32	&	3	&	58.56	&	56.19	&	4\%	&	167.3	&	221.0	&	24\%	\\
48	&	3	&	58.56	&	56.09	&	4\%	&	170.1	&	221.9	&	23\%	\\
\hline
\end{tabular}
\label{jumbo10gbeth}
\end{table}

\end{document}